\ifx\ipeinclude\undefined
\documentclass[11pt,a4paper]{article}
\else\fi
\usepackage[cp1250]{inputenc}
\usepackage[IL2]{fontenc}

\usepackage{color}

\usepackage{a4wide}
\usepackage{amssymb,amsmath,amsthm}
\usepackage{bbm}
\usepackage[bookmarks=false,pdfstartview={FitH},colorlinks=true,citecolor=blue,urlcolor=blue,linkcolor=blue,unicode=true]{hyperref} 

\usepackage{graphicx}
\usepackage{xspace}
\usepackage{ifthen}
\newboolean{MareksFormat}
\InputIfFileExists{MareksSwitch}{}{}
\ifthenelse{\boolean{MareksFormat}}{
 \usepackage{geometry}
\geometry{papersize={145mm,180mm},top=2.5mm,bottom=2.5mm,left=1.5mm,right=1.5mm}
}{}

\usepackage{tikz}
\usetikzlibrary{matrix,arrows}
\usepackage[matrix,arrow,cmtip,rotate,curve,arc]{xy}
\SelectTips{cm}{}
\input xy
\xyoption{all}
\newlength{\hlp}
\newcommand{\leftbox}[2]{\settowidth{\hlp}{$#1$}\makebox[\hlp][l]{${#1}{#2}$}}
\newcommand{\rightbox}[2]{\settowidth{\hlp}{$#2$}\makebox[\hlp][r]{${#1}{#2}$}}

\theoremstyle{plain}
\newtheorem{theorem}{Theorem}[section]
\newtheorem{lemma}[theorem]{Lemma}
\newtheorem{prop}[theorem]{Proposition}
\newtheorem{corol}[theorem]{Corollary}

\theoremstyle{definition}
\newtheorem{defi}[theorem]{Definition}
\newtheorem{remark}[theorem]{Remark}

\DeclareMathOperator{\Cyl}{Cyl}
\DeclareMathOperator{\Cone}{Cone}

\def\dim{\mathop{\rm dim}\nolimits}

\newcommand{\N}{\mathbbm{N}}
\newcommand{\Z}{\mathbbm{Z}}
\newcommand{\Q}{\mathbbm{Q}}

\def\id{\mathrm{id}}

\def\ndg{\textup{ndg}}

\def\bs{\mbox{\boldmath$\sigma$}}

\def\bt{\mbox{\boldmath$\tau$}}

\newcommand\makevec[1]{{\boldsymbol{#1}}}
\def\aa{\makevec{a}}
\def\bb{\makevec{b}}
\def\cc{\makevec{c}}
\def\xx{\makevec{x}}

\def\ra{\rightarrow}
\renewcommand\:{\colon}
\makeatletter
\newcommand{\ProofEndBox}{{\ifhmode\unskip\nobreak\hfil\penalty50 \else
          \leavevmode\fi\quad\vadjust{}\nobreak\hfill$\Box$
                      \finalhyphendemerits=0 \par}}%
\makeatother
\newcommand{\proofend}{\ProofEndBox\smallskip}

\newlength{\fparwidth}
\newcommand{\iipefig}[1]  
{\immediateFigure{\IPE{#1.ipe}}}
\def\immediateFigure#1{%
\smallskip\begin{center}#1\end{center}\smallskip }

\long\def\onefigure#1#2{
\begin{figure*}[tbp]
\begin{center}
#1
\end{center}
\caption{#2}
\end{figure*}
}
\def\immediateFigure#1{%
\smallskip\begin{center}#1\end{center}\smallskip }
\newcommand{\labfig}[2]  
{\onefigure{\mbox{\includegraphics{#1}}}{\label{f:#1} #2} }

\newcommand{\labfigw}[3]  
{\onefigure{\mbox{\includegraphics[width=#2]{#1}}}{\label{f:#1}
#3}}

\newcommand{\immfig}[1]  
{\immediateFigure{\mbox{\includegraphics{#1}}}}

\newcommand{\immfigw}[2] 
{\immediateFigure{\mbox{\includegraphics[width=#2]{#1}}}}

\def\indef#1{\emph{#1}}

\newcommand{\heading}[1]{\vspace{1ex}\par\noindent{\bf\boldmath #1}}

\newboolean{Comments}
\setboolean{Comments}{false}

\ifthenelse{\boolean{Comments}}{
\newcommand{\marrow}{\marginpar{\boldmath$\longleftarrow$}}
\newcommand{\marek}[1]{\ifhmode\newline\fi\marrow \textsf{\color{orange}*** (MAREK: ) #1\newline}}
\newcommand{\martin}[1]{\ifhmode\newline\fi\marrow \textsf{\color{red}{*** (MARTIN: ) #1\newline}}}
\newcommand{\jirka}[1]{\ifhmode\newline\fi\marrow \textsf{\color{magenta}*** (JIRKA: ) #1\newline}}
\newcommand{\uli}[1]{\ifhmode\newline\fi\marrow \textsf{\color{cyan}{*** (ULI: ) #1\newline}}}
\newcommand{\lukas}[1]{\ifhmode\newline\fi\marrow \textsf{\color{red}{*** (LUKAS: ) #1\newline}}}

\usepackage{soul}
\setstcolor{orange}
\setulcolor{orange}
\setul{}{0.3ex}
\def\del#1{\st{#1}}

\newenvironment{insbl}[1][]{\marek{\ul{#1 The following block will be INSERTED}}}
{\marek{\ul{End of INSERTED block}}}
}{
\newcommand{\marrow}{}
\newcommand{\marek}[1]{}
\newcommand{\uli}[1]{}
\newcommand{\jirka}[1]{}
\newcommand{\martin}[1]{}
\newcommand{\lukas}[1]{}
\def\del#1{}

}

\def\kamsymb{{\rm b}}
\def\ethsymb{{\rm c}}
\def\massymb{{\rm a}}
\def\epflsymb{{\rm d}}

\usepackage{color}
\newcommand\thedim{k}
\newcommand{\sddim}{{t}}
\newcommand{\then}{{n}}
\newcommand{\sphdim}{{p}}
\newcommand{\thelen}{{m}}
\newcommand{\theeqn}{{q}}

\DeclareMathOperator{\sd}{Sd}
\newcommand{\sdit}{\sd^{\sddim}}
\DeclareMathOperator{\size}{\mathsf{size}}
\DeclareMathOperator{\Ex}{\mathcal{K}}

\newcommand{\anick}[1]{Y^{4}_{#1}}
\newcommand{\rcone}{\mathop\mathrm{\widetilde{\Cone}}\nolimits}
\newcommand{\rcyl}{\mathop\mathrm{\widetilde{\Cyl}}\nolimits}
\newcommand{\skel}[2]{#2^{(#1)}}
\newcommand{\lastvertex}{\mathop\mathrm{lastv}\nolimits}
\renewcommand{\circ}{}

\newcommand{\chsimp}[2]{
D^{#1}\!\!\raisebox{-0.5ex}{\begin{tikzpicture}\draw (0ex,0ex) -- (0.8ex,1ex);\end{tikzpicture}}_{#2}}

\newcommand{\gmc}[1]{%
\settowidth{\hlp}{$\scriptstyle#1$}
\addtolength{\hlp}{10pt}
\ifthenelse{\lengthtest{\hlp<15pt}}{\setlength{\hlp}{15pt}}{}
\xymatrix@1@C=\hlp{\ar@{~>}[r]^-{#1}&}}

\newdir{c}{{}*!/-5pt/@^{(}}
\newdir{d}{{}*!/-5pt/@_{(}}
\newcommand{\embed}[1]{\xymatrix@1@C=15pt{{}\ar@{c->}[r]^{#1} & {}}}

\title{Extendability of continuous maps is undecidable\thanks{
This research was supported by the  ERC Advanced Grant No.~267165. The
research of M.~\v{C}.~was supported by the project
CZ.1.07/2.3.00/20.0003 of the Operational Programme Education
for Competitiveness of the Ministry of Education, Youth and
Sports of the Czech Republic. The research by M.\,K.\ and J.\,M.\ was
supported by the Center of Excellence -- Inst.\ for Theor.\
        Comp.\ Sci., Prague (project P202/12/G061 of GA~\v{C}R).
The research of L.~V.~was supported by the Center of Excellence -- Eduard \v{C}ech Institute (project P201/12/G028 of GA~\v{C}R).
The research by U.\,W.\ was supported by the Swiss National
Science Foundation (grants SNSF-200020-138230 and
SNSF-PP00P2-138948).}
}
\author{Martin \v{C}adek$^\massymb$ \and Marek Kr\v{c}\'al$^{\kamsymb}$ \and
Ji\v{r}\'{\i} Matou\v{s}ek$^{\kamsymb,\ethsymb}$
\and Luk\'a\v{s}
Vok\v{r}\'{\i}nek$^\massymb$
 \and Uli Wagner$^{\epflsymb}$}

\begin{document}
\maketitle
 {\renewcommand\thefootnote{\massymb} \footnotetext{Department
  of Mathematics and Statistics,
  Masaryk University, Kotl\'a\v{r}sk\'a~2, 611~37~~Brno,
  Czech Republic}}
 {\renewcommand\thefootnote{\kamsymb} \footnotetext{Department
  of Applied Mathematics,
  Charles University, Malostransk\'{e} n\'{a}m.~25,
  118~00~~Praha~1,  Czech Republic} }
 {\renewcommand\thefootnote{\ethsymb}
  \footnotetext{Institute of  Theoretical Computer Science, ETH
  Zurich, 8092~Zurich, Switzerland} }
  {\renewcommand\thefootnote{\epflsymb}
\footnotetext{Institut de Math\'{e}matiques de G\'{e}om\'{e}trie et Applications, {\'E}cole
Polytechnique F\'{e}d\'{e}rale de Lausanne, EPFL SB MATHGEOM, MA C1 553, Station 8, 1015
Lausanne, Switzerland}}

\begin{abstract}
We consider two basic problems of algebraic topology, the \emph{extension problem}
and the \emph{computation of higher homotopy groups}, from the point of view of computability
and computational complexity.

The \emph{extension problem} is the following: Given topological spaces $X$ and $Y$, a subspace $A\subseteq X$, and a (continuous) map $f\:A\to Y$, decide whether $f$ can be extended to a continuous map $\bar{f}\:X\to Y$. All spaces are given as finite simplicial complexes and the map $f$ is simplicial.

Recent positive algorithmic results, proved in a series of companion papers, show that for $(k-1)$-connected $Y$, $k\ge 2$,
the extension problem is algorithmically solvable if the dimension of $X$
is at most $2k-1$, and even in polynomial time when $k$ is fixed.

Here we show that the condition $\dim X\leq 2k-1$ cannot be relaxed: for $\dim X=2k$, the extension problem with $(k-1)$-connected $Y$
becomes undecidable. Moreover, either the target space $Y$ or
the pair $(X,A)$ can be fixed in such a way that the problem remains
undecidable.

Our second result, a strengthening of a result of Anick, says that the computation of $\pi_k(Y)$ of a $1$-connected simplicial complex $Y$ is \#P-hard when $k$ is considered as a part of the input.
\end{abstract}

\section{Introduction}
One of the central themes in algebraic topology is to understand
the structure of all \emph{continuous} maps
$X\to Y$, for given topological spaces $X$ and $Y$ (all maps
between topological spaces in this paper are assumed to be
continuous). For topological purposes, two maps $f,g\:X\to Y$
are usually considered equivalent if they are
\emph{homotopic}, i.e., if one can be continuously deformed
into the other\footnote{More precisely, $f$ and $g$ are
homotopic, in symbols $f\sim g$, if there is a
map $F\colon X\times [0,1]\to Y$ such that
$F(\cdot,0)=f$ and $F(\cdot,1)=g$. With this notation,
$[X,Y]=\{[f]: f\colon X\to Y\}$, where $[f]=\{g: g\sim
f\}$ is the \emph{homotopy class} of $f$.}; thus, the object
of interest is $[X,Y]$, the set of all homotopy classes of
maps $X\to Y$.

Many of the celebrated results throughout the history of
topology can be cast as information about $[X,Y]$ for
particular spaces $X$ and~$Y$. In   particular,  one of the important
challenges propelling the research in algebraic topology has been
the computation of the \emph{homotopy groups of spheres}\footnote{
We recall that the $\thedim$th homotopy group $\pi_\thedim(Y)$ of a space $Y$ is defined as the set of all
homotopy classes of \emph{pointed} maps $f\:S^\thedim\to Y$, i.e.,
maps $f$ that send a distinguished
\emph{basepoint} $s_0\in S^\thedim$ to a distinguished basepoint $y_0\in Y$ (and the homotopies $F$ also
satisfy $F(s_0,t)=y_0$ for all $t\in[0,1]$). Strictly speaking, one should
write $\pi_\thedim(Y,y_0)$ but for a path-connected $Y$, the choice of
$y_0$ does not matter. Moreover, if $Y$ is \emph{simply connected}, i.e., if $\pi_1(Y)$ is trivial,
then the pointedness of the maps does not matter either and one can identify $\pi_\thedim(Y)$ with
$[S^\thedim,Y]$. For $\thedim\geq 1$,
 each $\pi_\thedim(Y)$ is a group, which for $\thedim\ge 2$ is Abelian;
the definition of the group operation will be reviewed in Section~\ref{s:homotopy-groups}.}
$\pi_\thedim(S^n)$, for which only partial results have been obtained in spite of an enormous
effort (see, e.g., \cite{Ravenel,Kochman}).

A closely related question is the
\emph{extension problem}: given $A\subset X$ and a map
$f\:A\to Y$, can it be extended to a map $X\to Y$? For
example, the famous \emph{Brouwer Fixed-Point Theorem}
can be re-stated as non-extendability of the identity map $S^n\to
S^n$ to the ball $D^{n+1}$ bounded by the sphere $S^n$. See
\cite{Steenrod:CohomologyOperationsObstructionsExtendingContinuousFunctions-1972}
for a very clear and accessible introduction to the extension problem, including further
examples and applications and covering the earlier developments until the late 1950s.


\heading{Computational homotopy theory.}
In this paper, we consider the (theoretical) \emph{computational complexity} of
homotopy-theoretic questions such as the extension problem, the homotopy
classification of maps, and the computation of homotopy groups. More precisely,
we prove hardness and undecidability results that complement recent positive
algorithmic results obtained in a series of companion papers \cite{CKMSVW11,pKZ1,polypost}.
To put our results into context, we first give more background.

By classical uncomputability results in topology (see, e.g., the survey
\cite{Soare:ComputabilityDifferentialGeometry-2004}), most of these problems are \emph{algorithmically
unsolvable} if we place no restriction on the space $Y$ 
Indeed, by a result of Adjan and of Rabin,
it is undecidable whether the fundamental group $\pi_1(Y)$ of a given
finite simplicial complex $Y$ is trivial, even if $Y$ is assumed to be $2$-dimensional.
The triviality of $\pi_1(Y)$ is equivalent to $[S^1,Y]$ having only one element,
represented by the constant map, and so $[S^1,Y]$ is uncomputable
in general. Moreover, by the Boone--Novikov theorem, it is undecidable
whether a given pointed map $f\:S^1\to Y$ is homotopic to a constant map,
and this homotopic triviality is equivalent to the extendability
of $f$ to the 2-dimensional ball $D^2$. Therefore, the extension problem
is undecidable as well.\footnote{For undecidability
results concerning numerous more loosely related topological
problems we refer to
\cite{Soare:ComputabilityDifferentialGeometry-2004,
NabutovskyWeinberger:AlgorithmicAspectsHomeomorphismProblem-1999,
NabutovskyWeinberger:AlgorithmicUnsolvabilityTrivialityProblemMultidimensionalKnots-1996} and references therein.}

In these results, the difficulty stems from the intractability of the fundamental
group of $Y$. Thus, a reasonable restriction is to assume that
$\pi_1(Y)$ is trivial (which in general cannot be tested, but in many
cases of interest it is known), or more generally, that $Y$ is
\emph{$\thedim$-connected}, meaning that $\pi_i(Y)$ is trivial for all
$i\leq\thedim$
(equivalently, every map $S^i\to Y$, $i\leq\thedim$,
can be extended to $D^{i+1}$).
A basic and important example
of a $(\thedim-1)$-connected space is the sphere~$S^\thedim$.

For a long time, the only positive result concerning the computation of $[X,Y]$ was that of Brown~\cite{Brown:FiniteComputabilityPostnikovComplexes-1957}, who showed that $[X,Y]$ is computable under the assumption that $Y$ is  $1$-connected \emph{and} that all the higher homotopy groups $\pi_k(Y)$, $2\leq k\leq\dim X$,
are \emph{finite} (the second assumption is rather strong and \emph{not} satisfied if $Y$ is
a sphere, for example). Brown also gave an algorithm that, given $k\geq 2$ and a finite $1$-connected simplicial complex $Y$, computes $\pi_k(Y)$.

In the 1990s, three independent collections of works appeared
with the goal of making various more advanced methods
of algebraic topology \emph{effective} (algorithmic):
by Sch\"on \cite{Schoen-effectivetop}, by Smith
\cite{smith-mstructures}, and by Sergeraert, Rubio, Dousson, and Romero
 (e.g., \cite{Sergeraert:ComputabilityProblemAlgebraicTopology-1994,RubioSergeraert:ConstructiveAlgebraicTopology-2002,RomeroRubioSergeraert,SergRub-homtypes};
also see \cite{SergerGenova} for an exposition).
New algorithms for computing higher homotopy groups
follow from these methods; see Real \cite{Real96} for an algorithm
based on Sergeraert et~al.

An algorithm that computes $\pi_\thedim(Y)$ for a given $1$-connected
simplicial complex $Y$ in polynomial time for every \emph{fixed}
$\thedim\ge 2$ was recently presented in \cite{polypost},
also relying on \cite{pKZ1} and on the methods of effective homology
developed earlier by Sergeraert et al.





The problem of computing $[X,Y]$ was addressed in \cite{CKMSVW11}, where
it was shown that its structure is computable assuming that
$Y$ is $(\thedim-1)$-connected and $\dim(X)\leq 2k-2$,
for some integer $\thedim\geq 2$. These assumptions are sometimes
summarized by saying that $X$ and $Y$ are \emph{in the stable range}.

As observed in \cite{polypost}, the methods of \cite{CKMSVW11} can also be
used to obtain an algorithmic solution of the extension problem.
Here $\dim X$ can even be $1$ beyond the stable
range\footnote{In the border case $\dim X=2\thedim-1$, the algorithm
just decides the existence of an extension, while
for $\dim X\leq2\thedim-2$ it also yields a
classification of \emph{all} possible extensions up to homotopy.};
thus, given finite simplicial complexes $A\subseteq X$ and $Y$
and a simplicial map $f\:A\to Y$, where $Y$ is $(\thedim-1)$-connected
and $\dim X\leq 2\thedim-1$, $\thedim\geq 2$,
it can be decided algorithmically whether $f$
can be extended to a continuous map $X\to Y$.
The algorithm again runs in polynomial time for $\thedim$ fixed, and
the same holds for the algorithm mentioned above for
computing $[X,Y]$ in the stable range.

\paragraph{New undecidability results.}
For the algorithms for homotopy classification and extendability, we have two types of assumptions: The first is that
the dimension of $X$ is suitably bounded in terms of the connectivity of $Y$ (in the stable range or at most one more).
This is essential for the algorithms to work at all.\footnote{We remark that the stable range assumption guarantees that
$[X,Y]$ has a canonical Abelian group structure, which we exploit heavily (for instance, it means that $[X,Y]$ has a finite
description even when it is an infinite set). In the special case $\pi_k(Y)\cong [S^k,Y]$, by contrast, the group structure has a different origin and is available for all dimensions $k$.}
The second assumption is that the relevant dimensional parameter $k$ is fixed, which guarantees that the algorithm runs in polynomial time.

Our main result is that for the extension problem, the first assumption
is necessary and sharp.
%

\begin{theorem}
\label{t:undecide} Let $\thedim \geq 2$ be fixed.
\begin{enumerate}
\item[\rm (a)] {\rm (Fixed target) }
There is a \emph{fixed} $(\thedim-1)$-connected finite simplicial complex $Y=Y_\thedim$ such that the following problem is algorithmically unsolvable: Given finite simplicial complexes $A\subseteq X$ with $\dim X=2\thedim$ and a simplicial map $f\:A\to Y$, decide whether there exists a continuous map $X\to Y$ extending $f$. For $\thedim$ even, we can take $Y_\thedim$ to be the sphere $S^\thedim$.
\item[\rm (b)] {\rm (Fixed source) }
There exist fixed finite simplicial complexes $A=A_\thedim$ and $X=X_\thedim$ with $A\subseteq X$ and $\dim X=2\thedim$ such that the following problem is algorithmically unsolvable: Given a $(\thedim-1)$-connected finite simplicial complex $Y$
and a simplicial map $f\:A\to Y$, decide whether there exists a continuous map $X\to Y$ extending~$f$.
\end{enumerate}
\end{theorem}

The theorem is stated
in terms of simplicial complexes since these are a standard input
model for topological spaces in computational topology that we assume may be most familiar to most readers.
For the purposes of our reductions, we actually work with \emph{simplicial sets} (see Section~\ref{s:simplicial-sets}),
which offer a more flexible, but still purely combinatorial, way of representing topological spaces. The simplicial sets
are then converted into simplicial complexes by a suitable subdivision.

When constructing $A$, $X$ and $Y$ as simplicial sets, we can furthermore ensure that $Y$ has a certain additional property,
namely that it is \emph{$(\thedim-1)$-reduced}, which provides an immediate certificate that $Y$
is $(\thedim-1)$-connected; this is proved in \ref{a:reduced}. Thus, in particular, the difficulty of the extension problem does not lie in verifying the $(\thedim-1)$-connectedness of $Y$.

While most of the previous undecidability results in topology
rely on the word problem in groups and its relatives,
our proof of Theorem~\ref{t:undecide} relies on undecidability
of \emph{Hilbert's tenth problem}, which is the solvability of
a system of polynomial Diophantine equations, i.e., the existence
of an integral solution of a system of the form
\begin{equation}\label{e:dioph}
p_i(x_1,\ldots,x_r)=0, \ \ \ i=1,2,\ldots,s,
\end{equation}
where $p_1,\ldots,p_s$ are $r$-variate polynomials with integer
coefficients. This problem is undecidable by a celebrated
result of Matiyasevich \cite{Matiyasevich-Diophantineness-1970},
building on earlier work by Davis, Putnam, and Robinson;
also see \cite{Matiyasevich:Hilbert-1993,Mazur:SurveyUndecidabilityNumberTheory-1994} for additional background and further references.

\heading{On the hardness of computing $[X,Y]$.} When $\dim X=2k$ and $Y$
is $(k-1)$-connected, we can no longer equip $[X,Y]$  with the group
structure of the stable range.
Thus, it is not clear in what sense the
the potentially infinite set $[X,Y]$ could be computed in general.
A natural computational problem in this setting is to decide whether
$|[X,Y]|>1$; in other words, whether there is a homotopically nontrivial map
$X\to Y$ for given simplicial complexes $X$ and $Y$ as above.

We can prove that this problem is NP-hard for every \emph{even} $k\ge 2$;
in order to keep this paper reasonably concise, the proof is
to be presented in the PhD. thesis of the second author.
The reduction is very similar
to that of Theorem~\ref{t:undecide}. We can show that the problem is
at least as hard as deciding the existence of a nonzero integral solution
of the quadratic system \eqref{e:qsym} defined in
Section~\ref{s:undecide-diophantine} below with all the constant terms
$b_q$ equal to zero.
This problem may well be undecidable, but as far as we know,
the best known lower bound is that of NP-hardness.

\heading{\#P-hardness.} Our second result concerns the problem of computing the higher homotopy groups
$\pi_\then(Y)\cong [S^\then,Y]$ for a simply connected space $Y$, if $\then$ is not considered fixed but
part of the input ($\then$ is given in unary encoding).\footnote{Note that with a unary encoding of $\then$, the size of input is
significantly larger than with a binary (or decimal encoding), and hence the hardness
result is correspondingly stronger.} Anick \cite{Anick-homotopyhard} proved that this problem is \#P-hard,\footnote{Somewhat
informally, the class of  \#P-hard problems consists of computational
problems that should return a natural number
(as opposed to YES/NO problems) and are at least as hard as
counting the number of all Hamiltonian cycles in a given graph,
or counting the number of subsets with zero sum for a given set of integers,
etc. These problems are clearly at least as hard as NP-complete
problems, and most likely even less tractable.}
where $Y$ can even be assumed to be a $4$-dimensional space.%
\footnote{Actually, the hardness already applies to the potentially easier
problem of computing the \emph{rational homotopy groups}
$\pi_\then(Y)\otimes \Q$; practically speaking, one asks
only for the rank of $\pi_\then(Y)$,
i.e., the number of direct summands isomorphic to $\Z$.}

However, Anick's hardness result has the following caveat: it assumes that the input space $Y$ is given in a very concise form, as a
\emph{cell complex} with the  degrees of the \emph{attaching maps}
encoded in binary (see Section~\ref{s:anick-complexes} for a review
of the construction). A straightforward way of converting this cell complex
to a simplicial complex yields a $4$-dimensional simplicial complex with an exponential number of simplices, which renders the hardness result meaningless for simplicial complexes.
In Section~\ref{s:proofs},  we provide a different way of converting
Anick's concise encoding of the input space $Y$ into a homotopy
equivalent\footnote{Spaces $X$ and $Y$ are homotopy equivalent if
there are maps $f\:X\to Y$ and $g\:Y\to X$ such that the compositions $fg$
and $gf$ are homotopic to
identities. From the point of view of homotopy theory, such
$X$ and $Y$ are indistinguishable and, in particular,
$\pi_k(X)=\pi_k(Y)$ for all $k\geq 0$.}
simplicial complex that can be constructed in polynomial time, and
in particular, has only polynomially many simplices.
This yields the following result:

\begin{theorem}\label{t:sharpP} It is \#P-hard to compute
the rank of $\pi_\then(Y)$ (i.e. the number of summands of
$\pi_\then(Y)$ isomorphic to $\Z$) for a given number $\then\in \N$ (encoded
in \emph{unary}) and a given simply connected $4$-dimensional simplicial complex $Y$.
\end{theorem}

\heading{Outline of the proof of Theorem~\ref{t:undecide}. }
We aim to present our results in a way that makes the statements of
the results and the main steps and ideas accessible while assuming
only a moderate knowledge of topology on the side of the reader.

We thus review a number of basic topological concepts and provide
proofs for various assertions and facts that may be rather elementary
for topologists. On the other hand, some of the proofs assume a slightly
stronger topological background, since reviewing every single notion and
fact would make the paper too lengthy.

For proving Theorem~\ref{t:undecide},
we present an algorithm that converts a given system of Diophantine equations
into an instance of the extension problem; i.e., it constructs simplicial
complexes $A$, $X$ and $Y$ and a map $f\:A\ra Y$ such that there is an extension
of $f$ to all of $X$ iff the given system of equations is solvable.
Moreover, as stated in the theorem, there are actually
two versions of the reduction: The first uses a fixed target
space $Y=Y_\thedim$ and encodes the equations into $A$, $X$, and $f$.
The second uses a fixed pair $(X_\thedim,A_\thedim)$
of source complexes and encodes the equations into $f$ and~$Y$.

We will actually work only with \emph{quadratic} Diophantine equations
of a slightly special form
(which is sufficient; see Section~\ref{s:undecide-diophantine}).
The unknowns are represented by the degrees of restrictions of
the desired extension $\bar f$ to suitable $\thedim$-dimensional
spheres. The quadratic terms in the equations are obtained
using the \emph{Whitehead product}, which is a binary
operation that, for a space $Z$, assigns to
elements $\alpha\in\pi_k(Z)$ and $\beta\in\pi_\ell(Z)$
an element $[\alpha,\beta]\in\pi_{k+\ell-1}(Z)$;
see Section~\ref{s:Whitehead}.

Here is a rough outline of the proof strategy. First we focus on
Theorem~\ref{t:undecide}~(a) (fixed target) with $\thedim$ even,
which is the simplest among our constructions.
\begin{itemize}
\item The spaces $X$ and $A$ are simplest to describe
 as  cell complexes.
The subcomplex $A$ is a union of $r$ spheres
$S^{2\thedim-1}$, which intersect only at a single common point.
This union is called a wedge sum and denoted by
$A=S^{2\thedim-1}\vee\cdots\vee S^{2\thedim-1}$.
The space $X$ is homotopy equivalent to another wedge sum,
of $s$ spheres $S^\thedim$; i.e.,
 $X\simeq S^\thedim\vee\cdots\vee S^\thedim$.

\item The fixed $(\thedim-1)$-connected target space $Y$ is
 the $\thedim$-sphere~$S^\thedim$.

\item Maps $X\to S^\thedim$ can be described completely by their restrictions
to the $\thedim$-spheres in the wedge sum.
Each such restriction is characterized, uniquely up to homotopy,
by its degree---this can be an arbitrary integer.
Thus, a potential extension $\bar f$ can be encoded into a vector
$\xx=(x_1,\ldots,x_r)$ of integers.

\item Similarly, the map $f\:A\to S^\thedim$ can be described by its restrictions to the $(2\thedim-1)$-spheres in the wedge sum.
Crucially for our construction, the homotopy group $\pi_{2\thedim-1}(S^\thedim)$ has an element of infinite order, namely, the Whitehead square
$[\iota,\iota]$, where $\iota$ is the identity $S^\thedim\to S^\thedim$
(we still assume $\thedim$ even). We will work with maps $f$
whose restriction to the $\theeqn$th sphere is (homotopic to)
an integral multiple $b_\theeqn[\iota,\iota]$, for some (unique) integer
$b_\theeqn$.
Thus, $f$ is specified by the vector $\bb=(b_1,\ldots,b_s)$ of
these integers.

\item Given arbitrary integers $a_{ij}^{(\theeqn)}$, $1\le i<j\le r$,
$\theeqn=1,2,\ldots,s$, we construct the pair $(X,A)$ in such a way that,
taking $\bar f\:X\to Y$ specified by $\xx$ as above,
the restriction of $\bar f$ to
the $q$th sphere of $A$ is homotopic to
$\sum_{i< j} a_{ij}^{(\theeqn)}x_ix_j[\iota,\iota]$
(here the addition and multiplication by integers are
performed in $\pi_{2\thedim-1}(S^\thedim)$).
Since $[\iota,\iota]$ is an element of
$\pi_{2\thedim-1}(S^\thedim)$ of infinite order,
$\bar f$ is an extension of $f$ iff
 $\sum_{i< j} a_{ij}^{(\theeqn)}x_ix_j=b_\theeqn$ for
all $\theeqn=1,2,\ldots,s$.

\item In this way, we can simulate an arbitrary system of quadratic equations
by an extension problem. Some more work is still needed
to describe $X$ and $A$ as finite simplicial complexes and $f$
as a simplicial map.

\item For Theorem~\ref{t:undecide}~(a) with $\thedim$ odd,
the Whitehead square $[\iota,\iota]$ as above no longer has infinite
order. Instead, we use $Y=S^\thedim\vee S^\thedim$ and
replace $[\iota,\iota]$ by the Whitehead product
$[\iota_1,\iota_2]$ of the inclusions of the two spheres
into~$Y$. This leads to skew-symmetric systems of quadratic equations,
and showing that these are still undecidable needs some work
(see~Section~\ref{s:undecide-diophantine}).
\end{itemize}

In Theorem~\ref{t:undecide}~(b) with $\thedim$ even, the (fixed) source
space $X$ is homotopy equivalent to $S^\thedim$ and $A=S^{2\thedim-1}$.
Under the homotopy equivalence $X\simeq S^\thedim$, the inclusion $A\hookrightarrow X$
becomes the Whitehead
square $[\iota,\iota]$, and for $\thedim$ odd it is replaced by
$[\iota_1,\iota_2]$. In both cases, the system of quadratic equations is
encoded into the structure of the cell complex $Y$ and the map $f\:A\to Y$.

\section{Diophantine equations and undecidability}
\label{s:undecide-diophantine}




We will need to work with quadratic Diophantine equations of two special forms:

{\renewcommand\theequation{Q-SYM}
\begin{equation}\label{e:qsym}
\sum_{1\leq i< j\leq r} a_{ij}^{(\theeqn)}x_i x_j=b_\theeqn,\ \ \ \theeqn=1,2,\ldots,s,
\end{equation}}%
where $a_{ij}^{(\theeqn)},b_\theeqn\in\Z$ and $x_1,\ldots,x_r$ are the
unknowns (i.e., the left-hand sides are quadratic
forms with no square terms), and {\renewcommand\theequation{Q-SKEW}
\begin{equation}\label{e:qskew}
\sum_{1\leq i<j\leq r} a_{ij}^{(\theeqn)} (x_iy_j-x_jy_i)= b_\theeqn,\ \ \ \theeqn=1,2,\ldots,s,
\end{equation}}%
with  $a_{ij}^{(\theeqn)},b_\theeqn\in\Z$ and unknowns
$x_1,\ldots,x_r$, $y_1,\ldots,y_r$ (so here we deal
with skew-symmetric bilinear forms).


\begin{lemma} The solvability
of  the system {\rm \eqref{e:qsym}}, as well as that of {\rm
\eqref{e:qskew}}, in the integers are algorithmically
undecidable.
\end{lemma}

\begin{proof} First, it is well known and easy to see that the solvability
of a general \emph{quadratic} system of Diophantine equations
is no easier than the solvability of an arbitrary Diophantine system
(\ref{e:dioph}), and thus undecidable.\footnote{
The idea is to represent higher-degree monomials in the
general system using new variables; e.g., for the monomial $x^3y$
we can introduce new variables $t_1,t_2$, new quadratic equations $t_1=x^2$
and $t_2=xy$, and replace $x^3y$ by $t_1t_2$.}

First we show undecidability for \eqref{e:qsym}; this system
differs from a general quadratic system
only by the lack of linear terms and squares.
Given a general quadratic system
\begin{equation}\label{e:GQ}
\sum_{1\leq i,j\leq r} a_{ij}^{(\theeqn)} x_ix_j+\sum_{1\leq i\leq r} b_{i}^{(\theeqn)}x_i= c_\theeqn,\ \ \ \theeqn=1,\ldots,s,
\end{equation}
we add new variables $x_0$, $x_0'$ and $x_1',\ldots,x_r'$,
and we replace the terms $x_ix_j$ with $x_ix_j'$ and $x_i$ with $x_ix_0'$.
We also add the following equations
\[x_0x_0'=1;\ \ \ x_ix_0'-x_0x_i'=0,\ \ \ i=1,\ldots,s.\]
The resulting system is of the form \eqref{e:qsym}
(assuming an indexing of the variables such that the $x_i$ precede the $x'_i$)
 and it forces $x_0=x_0'=\pm1$, and $x_i=x_i'$.
Thus, each of its solutions corresponds
either to a solution of the original system
\eqref{e:GQ}
(when $x_0=x_0'=1$), or to a solution of
the system obtained from \eqref{e:GQ} by changing the sign of all
the linear terms (when $x_0=x_0'=-1$). Since
 there is an obvious bijection $x_i\mapsto -x_i$ between the solutions of
\eqref{e:GQ} and those of the system with negated linear terms,
 the solvability of the constructed system \eqref{e:qsym} is equivalent
to the solvability of~\eqref{e:GQ}.
\medskip

Next, we show that \eqref{e:qskew} is no easier than \eqref{e:qsym}.
Given a general system \eqref{e:qsym},
 we add new variables $x_0$, $y_0$, $x_0'$, $y_0'$ and, for each $i=1,\ldots,r$, also $x_i'$, $y_i$ and $y_i'$. We replace
each term $x_ix_j$ in the original system \eqref{e:qsym} by the antisymmetric
expression $x_iy_j'-x_j'y_i$, and we add the following equations
(for $i=1,2,\ldots,r$):
\[x_0y_0'-x_0'y_0=1,\ \ \ x_0y_i-x_iy_0=0,\ \ \ x_0'y_i'-x_i'y_0'=0,\ \ \ (x_0y_i'-x_i'y_0)-(x_iy_0'-x_0'y_i)=0.\]
This gives a system of the form \eqref{e:qskew}, which we call
the \emph{new system}.

It is clear that each solution of \eqref{e:qsym} yields a solution of
the new system. Conversely, supposing that the new system has
a solution, we claim that it also has a solution with
$x_0=y_0'=1$ and $y_0=x_0'=0$. Once we have a solution satisfying
these additional conditions, it is easy to check that $x_1,\ldots,x_r$
form a solution of the original system.

To verify the claim, for notational convenience, let us index the $x$ and $y$
variables in the new system by the set $I=\{0,1,\ldots,r,0',1',\ldots,r'\}$,
where $x_{i'}=x'_i$ and $y_{i'}=y'_i$. We suppose that
$(x_i,y_i:i\in I)$ form a solution of the new system.
Since $x_0y_0'-x_0'y_0=1$, the $2\times 2$ matrix
$\left(\begin{smallmatrix} x_0 & x_0' \\ y_0 & y_0' \end{smallmatrix}\right)$
has determinant $1$ and thus an integral inverse matrix, which we denote by~$T$.

Let us define new values $(\bar x_i,\bar y_i:i\in I)$ by
$\left(\begin{smallmatrix} \bar x_i \\ \bar y_i \end{smallmatrix}\right)=
T\cdot\left(\begin{smallmatrix} x_i \\ y_i \end{smallmatrix}\right)$, $i\in I$.
We have $\bar x_0=\bar y_0'=1$ and $\bar y_0=\bar x_0'=0$,
and it remains to show that the $\bar x_i$ and $\bar y_i$ satisfy the new
system. This is because, for every $i,j\in I$, we have
\begin{eqnarray*}
\bar x_i\bar y_j-\bar x_j\bar y_i&=&
\det\begin{pmatrix} \bar x_i & \bar x_j \\ \bar y_i & \bar y_j \end{pmatrix}
=\det\left(T\cdot\begin{pmatrix} x_i & x_j \\ y_i & y_j \end{pmatrix}\right)\\
&=&\det T\cdot\det\begin{pmatrix} x_i & x_j \\ y_i & y_j \end{pmatrix}
=x_iy_j-x_jy_i.
\end{eqnarray*}
\end{proof}

\section{Cell complexes and simplicial sets}
\label{s:simplicialCW}

This section and the next one mostly present known material from topology;
in several cases we need to adapt results from the literature to our needs,
which is sometimes best done by re-proving them. Readers may want to skim
these two sections quickly and return to them later when needed.

Here we review two basic ways
of building topological spaces from simple pieces:
\emph{cell complexes} and \emph{simplicial sets}.
Cell complexes, also known as \emph{CW complexes}, are fairly standard
in topology, and we will use them for a simple description of the
various spaces in our proofs. Simplicial sets are
perhaps less well known, and for us, they will mainly be
a convenient device for converting cell complexes into simplicial
complexes. Moreover, they are of crucial importance in the algorithmic
results mentioned in the introduction.
For a thorough discussion of simplicial complexes, simplicial sets,
cell complexes, and the connections between the three,
we refer to \cite{FritschPiccinini:CellularStructures-1990}.

\subsection{Cell complexes}
\label{s:cell-complexes}

In the case of cell complexes, the building blocks are topological disks of various dimensions,
called \emph{cells}, which can be thought of as being completely ``flexible'' and which can be
glued together in an almost arbitrary continuous fashion. Essentially the only condition is that each $n$-dimensional
cell has to be attached along its boundary to the \emph{$(n-1)$-skeleton} of the space, i.e., to the part
that has already been built, inductively, from lower-dimensional cells.
The formal definition is as follows.

We recall that if $X$ and $Y$ are topological spaces and if $f\:A\rightarrow Y$ is a map defined on a subspace $A\subseteq X$, then the space $X\cup_f Y$ obtained by \emph{attaching $X$ to $Y$ via $f$} is defined as
the quotient of the disjoint union $X\sqcup Y$ under the equivalence relation generated by the identifications $a\sim f(a)$, $a\in A$.

A \emph{closed} or \emph{open} $n$-\emph{cell} is a space homeomorphic to the
closed $n$-dimensional unit disk $D^n$ in $n$-dimensional Euclidean space or its
interior $\mathring{D}^n$, respectively; a point is regarded as both a closed and an open
$0$-cell.

An \emph{$m$-dimensional cell complex}\footnote{Cell complexes can be also infinite-dimensional, in which case some care has to be taken in defining their topology, but we will deal with cell complexes made of finitely many cells, and thus finite-dimensional.} $X$ is the last term of an inductively constructed sequence of spaces $\skel{0}{X}\subseteq \skel{1}{X}\subseteq \skel{2}{X}\subseteq \ldots\subseteq \skel{m}{X}=X$, called the \emph{skeleta} of $X$:

\begin{enumerate}
\item $X^{(0)}$ is a discrete set of points (possibly infinite) that are regarded as $0$-cells.
\item Inductively, the $n$-skeleton $\skel{n}{X}$ is formed by attaching closed
$n$-cells $D^n_i$ (where $i$ ranges over some arbitrary index set) to $\skel{n-1}{X}$ via
\emph{attaching maps} $\varphi_i\:S^{n-1}_i=\partial D^n \to \skel{n-1}{X}$.
Formally, we can consider all attaching maps together as defining a map $\varphi=\sqcup_i\varphi_i$
from the disjoint union $\bigsqcup_i S_i^{n-1}$ to $\skel{n-1}{X}$ and form
$\skel{n}{X}=\big(\bigsqcup_iD_i^n\big) \cup_\varphi \skel{n-1}{X}$.
%
%
\end{enumerate}

For every closed cell $D_i^n$, one has a \emph{characteristic map}\footnote{The composition of the inclusion $D_i^n \hookrightarrow \big(\bigsqcup_iD_i^n\big)\sqcup \skel{n-1}{X}$ with the quotient map.} $\Phi_i\:D_i^n \to \skel{n}{X}\subseteq X$, which restricts to an embedding on the interior $\mathring{D}_i^n$. The image $\Phi_i(\mathring{D}_i^n)$
is commonly denoted by $e_i^n$, and it follows from the construction that every point of $X$ is contained in a unique open cell (note that these are in general not open subsets of $X$, however).

As a basic example, the $n$-sphere is a cell complex with one $n$-cell and one $0$-cell, obtained by attaching $D^n$ to a point $e^0$ via the constant map that maps all of $S^{n-1}$ to $e_0$.


\heading{Subcomplexes
.}
A \emph{subcomplex} $A\subseteq X$ is a subspace that is closed and a union of open cells of $X$.
In particular, for each cell in $A$, the image of its attachment map is contained in $A$, so $A$ is itself
a cell complex (and its cell complex topology agrees with the subspace topology inherited from $X$).


\heading{The homotopy extension property.} An important fact is
that cell complexes
 have the so-called \emph{homotopy extension property}:
Suppose that $X$ is a cell complex and that $A\subseteq X$ is a subcomplex.
If we are given a map $f_0\:A\to Y$ into a some space $Y$, an extension
$\bar{f}_0\:X\to Y$ of $f_0$ and a homotopy $H\:A\times [0,1]$ between $f_0$
and some other map $f_1\:A\to Y$, then $H$ can be extended to a homotopy
$\bar{H}\:X\times [0,1]\to Y$ between $\bar{f}_0$ and some extension
$\bar{f}_1\:X\to Y$ of $f_1$. Here is an immediate consequence:

\begin{corol}\label{c:ext} For a cell complex $X$, subcomplex
$A\subseteq X$, and a space $Y$,
the extendability of a map $f\:A\to Y$ to $X$ depends only  on the homotopy class of $f$
in $[A,Y]$.
Moreover, the map $f\:A\to Y$ has an extension $\bar{f}\:X\to Y$ iff there exists a map $g\:X\to Y$ such that the diagram
$$
\xymatrix{
A\ar[r]^f  \ar[d]_i & Y\\
X\ar[ru]_g &
}$$
commutes up to homotopy, i.e., $gi\sim f$.
\end{corol}

\heading{Cellular maps and cellular approximation.} A map $f\:X\to Y$ between cell complexes is called \emph{cellular} if it maps skeleta to skeleta, i.e., $f(\skel{n}{X})\subseteq \skel{n}{Y}$ for every $n$.

The \emph{cellular approximation theorem} (see \cite[Thm.~4.8]{Hatcher}) states that every continuous map
$f\:X\to Y$ between cell complexes is homotopic to a cellular one; moreover, if the given map $f$ is already cellular on some subcomplex $A\subseteq X$, then the homotopy can be taken to be \emph{stationary} on $A$ (i.e., the image of every point in $A$ remains fixed throughout).

\subsection{Simplicial sets}
\label{s:simplicial-sets}

For certain constructions it is advantageous to use a special type of
cell complexes with an additional structure that allows for a purely combinatorial description; the latter
also facilitates representing and manipulating
the objects in question, simplicial sets, on a computer.
We refer to \cite{Friedm08} for a very friendly thorough introduction to simplicial sets.

Intuitively, a simplicial set can be thought of as a kind of hybrid or compromise between a simplicial complex (more special) on the
one hand and a cell complex (more general) on the other hand. Like in the case of simplicial
complexes, the building blocks (cells) of which a simplicial set is constructed are simplices
(vertices, edges, triangles, tetrahedra, \ldots), and the boundary of each $n$-simplex $\Delta^n$ is attached
to the lower-dimensional skeleton by identifications that are linear on each proper face (subsimplex) of $\Delta^n$; thus, these identifications can be described combinatorially by maps between the vertex sets
of the simplices.\footnote{More precisely, the vertex set of each simplex is equipped with an ordering, and the identifications are required to be weakly order-preserving maps (not necessarily injective) between the vertex sets.} However, the attachments are more general than the one permitted for simplicial complexes; for example, one may have several
1-dimensional simplices connecting the same pair of vertices,
a 1-simplex forming a loop, two edges of a 2-simplex
identified to create a cone, or the boundary of a 2-simplex
all contracted to a single vertex, forming an $S^2$.
\immfig{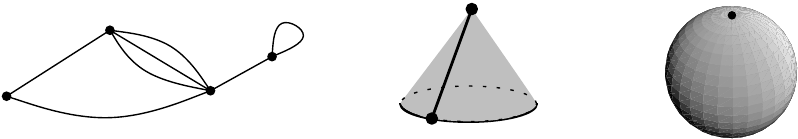}
 Moreover, one keeps track of certain additional information that might seem superfluous but turns
 out to be very useful for various constructions. For instance, even if the identifications force some
 $n$-simplex to be collapsed to something lower-dimensional (so that it could be discarded for the purposes
of describing the space as a cell complex), it will still be formally kept on record as a \emph{degenerate}
$n$-simplex; for instance, the edges of the triangle with a boundary contracted to a point (the last example above) do not disappear---formally, each of them keeps a phantom-like
existence of a degenerate $1$-simplex.

Formally, a simplicial set $X$ is given by a sequence
$(X_0,X_1,X_2,\ldots)$ of mutually disjoint sets, where the
elements of $X_\then$ are called the \emph{$\then
$-simplices of $X$} (we note that, unlike for simplicial
complexes, a simplex in a simplicial set need not be
determined by the set of its vertices; indeed, there can be
many simplices with the same vertex set). The $0$-simplices
are also called \emph{vertices}.

For every $\then \ge 1$, there are $\then +1$ mappings
$\partial_0,\ldots,\partial_\then\:X_\then\to
X_{\then-1}$ called \indef{face operators}; the intuitive
meaning is that for a simplex $\sigma\in X_\then$,
$\partial_i\sigma$ is the face of $\sigma$ opposite to the
$i$th vertex. Moreover, there are $\then +1$ mappings
$s_0,\ldots,s_\then\:X_\then\to X_{\then+1}$ called
the \emph{degeneracy operators}; the approximate meaning of $s_i\sigma$ is the degenerate simplex
which is geometrically identical to $\sigma$, but with the
$i$th vertex duplicated. A simplex is called
\indef{degenerate} if it lies in the image of some $s_i$;
otherwise, it is \indef{nondegenerate}. We write $X^\ndg$ for
the set of all nondegenerate simplices of~$X$.
A simplicial set
is called \emph{finite} if it has only finitely many nondegenerate
simplices (if $X$ is nonempty, there are always infinitely many
degenerate simplices, at least one for every positive dimension).

There are natural axioms that the $\partial_i$
and the $s_i$ have to satisfy, but we will not list them here,
since we won't really use them. Moreover, the usual definition
of simplicial sets uses the language of category theory
and is very elegant and concise; see, e.g., \cite[Sec.~4.2]{FritschPiccinini:CellularStructures-1990}.

If $A$ and $X$ are simplicial sets such that $A_n\subseteq X_n$ for every $n$ and the face and degeneracy
operators of $A$ are the restrictions of the corresponding operators of $X$, then we call $A$ a \emph{simplicial subset}
of $X$.

\heading{Examples.} Here we sketch some basic examples of
simplicial sets; again, we won't provide all details,
referring to \cite{Friedm08}. Let $\Delta^{\sphdim}$ denote the
standard $\sphdim$-dimensional simplex regarded as a simplicial
set. For $\sphdim=0$, $(\Delta^0)_\then$ consists of a single simplex,
denoted by $0^{\then}$,  for every $\then=0,1,\ldots$;  $0^0$ is the
only nondegenerate simplex. The face and degeneracy operators
are defined in the only possible way.

For $\sphdim=1$, $\Delta^1$ has two $0$-simplices (vertices), say
$0$ and $1$, and in general there are $\then+2$ simplices in
$(\Delta^1)_\then$; we can think of the $i$th one as containing
$i$ copies of the vertex $0$ and $\then+1-i$ copies of the vertex
$1$, $i=0,1,\ldots,\then+1$. For $\sphdim$ arbitrary, the $\then$-simplices
of $\Delta^\sphdim$ can be thought of as all nondecreasing
$(\then+1)$-term sequences with entries in $\{0,1,\ldots,\sphdim\}$;
the ones with all terms distinct are nondegenerate.

In a similar fashion, every simplicial complex $K$ can be
converted into a simplicial set $X$ in a canonical way;
first, however, we need to fix a linear ordering of the
vertices. The nondegenerate $\then$-simplices of $X$ are in
one-to-one correspondence with the $\then$-simplices of $K$, but
many degenerate simplices show up as well.

\heading{Geometric realization.}
Like a simplicial complex, every simplicial set $X$ defines a topological space
$|X|$, the \indef{geometric realization of $X$}, which is unique up to homeomorphism.
More specifically, $|X|$ is a cell complex with one $n$-cell for every \emph{nondegenerate}
 $n$-simplex of $X$, and these cells are glued together according to the identifications
implied by the face and degeneracy operators (we omit the precise definition of the attachments,
since we will not really use it and refer to the literature, e.g., to \cite{Friedm08} or \cite[Sec.~4.3]{FritschPiccinini:CellularStructures-1990}).

\heading{Simplicial maps.}
Simplicial sets serve as a combinatorial way of describing a topological
space; in a similar way, simplicial maps provide a combinatorial description of
continuous maps.

A \indef{simplicial map} $f\:X\to Y$ of simplicial sets $X,Y$
consists of maps $f_\then\:X_\then\to Y_\then$, $\then=0,1,\ldots$, that
commute with the face and degeneracy operators.

A simplicial map $f\:X\to Y$ induces a continuous, in fact, a cellular map
$|f|\:|X|\to|Y|$ of the geometric realizations in a natural
way (we again omit the precise definition). Often we will
take the usual liberty of omitting $|\cdot|$ and not
distinguishing between simplicial sets and maps and their
geometric realizations.

Of course, not all continuous maps are induced by simplicial
maps. However, simplicial maps can be used to \emph{approximate}
arbitrary continuous maps up to homotopy.
The \emph{simplicial approximation theorem} (which
may be most familiar in the context of simplicial complexes) says that for an arbitrary
continuous map $\varphi\:|X|\to|Y|$ between the geometric realizations of
simplicial sets, with $X$ finite, there exist a \emph{sufficiently fine subdivision} $X'$ of
$X$ and a simplicial map $f\:X'\to Y$ whose geometric realization is homotopic to $\varphi$;
see Section~\ref{s:simplicial-approx} for more details.

\heading{Encoding finite simplicial sets.} A finite simplicial complex can be encoded in a
straightforward way by listing the vertices of each simplex.

For simplicial sets, the situation is a bit more complicated, since the simplices are no longer
uniquely determined by their vertices, but if $X$ is finite,
then we can encode $X$ by the set $X^\ndg$
of its nondegenerate simplices (which we assume to be numbered
from $1$ to $N$, where $N$ is the total number
of nondegenerate simplices), plus a little bit of additional information.

The simple but crucial fact (see, e.g. \cite[Thm.~4.2.3]{FritschPiccinini:CellularStructures-1990}) we need
is that every simplex $\sigma$ can be written uniquely as  $\sigma=s\tau$, where $\tau$ is nondegenerate and
$s$ is a degeneracy, i.e., a composition $s=s_{i_k}\circ \ldots \circ s_{i_1}$ of degeneracy operators
where $k=\dim \sigma-\dim \tau$ (in particular, $\sigma$ is nondegenerate itself if $\sigma=\tau$ and $s$ is the
identity). Thus, as mentioned above, degenerate simplices $\sigma$ do not need to be encoded explicitly but
can be represented by $s\tau$ when needed, where the degeneracy $s$ can be encoded by the sequence $(i_k,\ldots,i_1)$ of indices of its components.\footnote{Moreover, this sequence is unique, by the simplicial set axioms that we have not specified, if one stipulates $i_k<\ldots<i_1$.} The extra information we need to encode $X$, in addition to the list of its nondegenerate simplices, is how these fit together. Specifically, for $\sigma \in X_n^\ndg$ and $0\leq i\leq n$, the
$i$th face can be written uniquely as $\partial_i\sigma \in X_{n-1}=s\tau$ with $\tau$ nondegenerate, and for each $\sigma$,
we record the $(n+1)$-tuple of pairs $(\tau,s)$.

Similarly, if $f\:X\to Y$ is a simplicial map between finite simplicial sets, then given the encodings of $X$ and $Y$,
we can encode $f$ by expressing, for each $\sigma\in X_n^\ndg$, the image $f(\sigma)=s\tau$, with $\tau\in Y_m^\ndg$ and recording the list of triples $(\sigma,\tau,s)$.

For a finite simplicial set $X$, we define $\size(X)$ as the number
of nondegenerate simplices.
If the dimension of $X$ is bounded by some number $d$,
then the number of bits in the encoding of $X$ described above
is bounded by $O(\size(X)\log\size(X))$, with the constant
of proportionality depending only on~$d$.

The notion of size will be a convenient tool that allows us to ensure that
our reductions can be carried out in polynomial time, without analyzing
the running time in complete detail, which we feel would be cumbersome
and not very enlightening.

More specifically, our reductions will be composed of a sequence of various basic constructions
of simplicial sets, which will be described in the next subsection.

For each of these basic constructions, it is straightforward to check\footnote{A notable exception are
subdivisions, for which we provide more detail in an appendix.} that when we apply them
to finite simplicial sets of \emph{bounded dimension}, both the running time of the construction
(the number of steps needed to compute the encoding of the output from the encoding of the input)
as well as the size of the output simplicial set are polynomial in the size of the input. Thus, to ensure
polynomiality of the overall reduction, it will be enough to take care that we combine only a polynomial
number of such basic constructions, that the size of every intermediate
simplicial set constructed during the reduction remains polynomial in the initial input, and that the dimension
remains bounded.


\subsection{Basic constructions}
\label{s:general-constructions}

In this subsection, we review several basic constructions for cell complexes and
simplicial sets. (One advantage of simplicial sets over simplicial
complexes is that various operations on topological spaces, in particular Cartesian
products and quotients, have natural counterparts for simplicial sets.
This is where the degeneracy operators and degenerate simplices
turn out to be necessary.) For more
details, we refer to \cite{Hatcher,FritschPiccinini:CellularStructures-1990}.


\heading{Pointed and $k$-reduced simplicial sets and cell complexes.} Several of the
constructions are defined for pointed spaces. We recall that a \emph{pointed space}
$(X,x_0)$ is a topological space $X$ with a choice of a distinguished point $x_0\in X$ (the
\emph{basepoint}). If $X$ is a cell complex or a simplicial set then we
will always assume that the basepoint to be a \emph{vertex} (i.e., a $0$-cell or $0$-simplex, respectively).
A \emph{pointed map} $(X,x_0)\to(Y,y_0)$ of pointed spaces (cell complexes,
simplicial sets) is a continuous (cellular, simplicial) map sending $x_0$ to $y_0$.
Homotopies of pointed maps are also meant to be pointed; i.e., they must keep the
image of the basepoint fixed. The reader may recall that, for example,
the homotopy groups $\pi_k(Y)$ are defined as homotopy
classes of pointed maps. The set of pointed homotopy classes  of pointed maps $X\to Y$ will be denoted by $[X,Y]_*$.

A simplicial set $X$ is called \emph{$k$-reduced},
$k \geq 0$, if it has a single
vertex and no nondegenerate simplices in dimensions $1$
through $k$. Similarly, a cell complex $X$ is $k$-reduced if it has a single vertex
and no cells of dimensions $1$ up to $k$. It is then necessarily
$k$-connected.

If $(Y,y_0)$ is a $0$-reduced cell complex (or simplicial set), then any cellular (or simplicial)
map from a pointed complex $(X,x_0)$ into $Y$ is automatically pointed. Moreover, if $Y$ is
$1$-reduced, then every homotopy is pointed, too, and
thus $[X,Y]$ is canonically isomorphic to $[X,Y]_*$.

\heading{Products. } If $X$ and $Y$ are cell complexes, then their
Cartesian product $X\times Y$ has a natural cell complex structure
whose $\then$-cells are products $e^p\times e^q$, where $p+q=\then$ and
$e^p$ and $e^p$ range over the $p$-cells of $X$ and the $q$-cells of $Y$, respectively.

Furthermore, if $X$ and $Y$ are simplicial sets then there is a formally
very simple way to define their product $X\times Y$: one sets
$(X\times Y)_\then:=X_\then\times Y_\then$ for every $\then$, and the
face and degeneracy operators work componentwise;
e.g., $\partial_i(\sigma,\tau):=(\partial_i\sigma,\partial_i\tau)$.
As one would expect from a good definition, the product of
simplicial sets corresponds to the Cartesian product of their
geometric realizations, i.e., $|X\times Y|\cong|X|\times |Y|$.%
\footnote{To be more precise, the above equality
holds literally, with the product topology on the right hand side, only under
suitable assumptions on $X$ and $Y$, e.g., if both $X$ and $Y$ have only
countably many simplices. In the general case, one has to interpret the product
$|X|\times |Y|$ differently, in the category of so-called \emph{$k$-spaces},
and the same subtlety arises for products of cell complexes,
see, e.g., the discussion in the respective appendices in \cite{FritschPiccinini:CellularStructures-1990,Hatcher}.
For the spaces we will encounter, however, this issue will not arise and the product will be the same as the usual product
of topological spaces.}
The apparent simplicity of the definition hides some intricacies,
though, as one can guess after observing that, for example, the
product of two 1-simplices is not a simplex---so the above
definition has to imply some canonical way of triangulating the product.

\begin{remark}
A pair $(s\sigma,t\tau)$ of degenerate simplices in the factors may yield a nondegenerate simplex
in the product, if the degeneracies $s$ and $t$ are composed of different degeneracy operators $s_i$.
However, $\dim(X\times Y)=\dim X +\dim Y$, so the product contains no nondegenerate simplices of
dimension larger than $\dim X+\dim Y$,
 and hence $\size(X\times Y)$ is at most $\size(X)\times \size(Y)$ times some
factor that depends only on the dimension%
\footnote{This follows from the fact about realizations mentioned above. Another way of seeing this is that if $\dim \sigma=p$, $\dim \tau=q$ and $\dim(s\sigma)=\dim(t\tau)=n> p+q$ then $s$ and $t$ involve
$n-p$ and $n-q$ degeneracy operators $s_i$ with $i\leq n$, respectively, so there must be a repetition since $n-p+n-q>n$.
Without further reflection, this immediately implies that $\size(X\times Y)\leq \size(X)\cdot \size(Y)\cdot (\dim X)!(\dim Y)!$.

In fact, the factor is only singly exponential in the dimensions. For instance, for a product $\Delta^p\times \Delta^q$ of two standard simplices, the vertices of $\Delta^p\times \Delta^q$ correspond to the grid points in
$\{0,\ldots,p\}\times \{0,\ldots,q\}$, and the non-degenerate $k$-simplices correspond
to subsets of size $k+1$ of the grid that are weakly monotone in both coordinates (weakly monotone
paths of length $k$). Thus, the number of non-degenerate simplices of full dimension $p+q$ equals
$\binom{p+q}{p}$, and the number of all non degenerate simplices is at most $4^{p+q}$, say. Thus,
$\size(X\times Y)\leq \size(X)\cdot\size(Y)\times 4^{\dim X+\dim Y}$, say.}
$\dim(X\times Y)$.

Moreover, if the dimensions are bounded, the product can be constructed in polynomial time.
\end{remark}

\heading{Quotients and attachments.}
If $X$, $Y$ and $A$ are cell complexes with $A\subseteq X$ and if $f\:A\to Y$ is a cellular map, then
the space $X\cup _f Y$ obtained by attaching $X$ to $Y$ along $f$ is also a cell complex in a natural
way (see, e.g., \cite[Sec.~2.3]{FritschPiccinini:CellularStructures-1990}). In particular
, $X/A$ is a cell complex, with cells corresponding to the cells of $X$ not contained in $A$, plus one additional $0$-cell (corresponding to the image of $A$ under the quotient map).

Similarly, if $X$ is a simplicial set and if $\sim$ is an equivalence relation on
each $X_n$ that is compatible with the face and degeneracy operators, then
the quotient $X/\sim$ is also a simplicial set. In particular, this includes \emph{simplicial attachments}
$X\cup_f Y$ of simplicial sets along a simplicial map $f\:A\to Y$ defined on a simplicial subset $A\subseteq X$,
and quotients $X/A$ by simplicial subsets. These constructions are compatible with geometric realizations.
i.e., e.g., $|X\cup_f Y|\cong |X|\cup_{|f|} |Y|$.

Moreover, the size of $X\cup_f Y$ is at most the size of $X$ plus the size of $Y$,
and in bounded dimension, the attachment can be constructed in polynomial time.


\paragraph{Wedge sum (or wedge product).}
If $X_1,\ldots, X_\thelen$ are pointed spaces,
 then their \emph{wedge sum} $X_1\vee \cdots \vee X_\thelen$
is simply the disjoint union of the $X_i$ with the basepoints identified (this is a very special type of
attachment). If the $X_i$ are cell complexes or simplicial sets, then so is their wedge sum.

Later we will need the following bijection:
\begin{equation}\label{e:wedge}
[X_1\vee X_2\vee \cdots\vee X_\thelen,Y]_*\xrightarrow\cong [X_1,Y]_*\times [X_2,Y]_*\times\dots\times
[X_\thelen,Y]_*
\end{equation}
where the components of this map are given by the restrictions
 to the respective $X_i$.

\heading{Mapping cylinder and mapping cone.} For a map $f\:X\ra Y$,
 the \emph{mapping cylinder} of $f$ is the space
$\Cyl(f)$ defined as the quotient of $(X\times [0,1])\sqcup Y$ under the identifications $(x,0)\sim f(x)$
for each $x\in X$.
The \emph{mapping cone} $\Cone(f)$ is defined as the quotient $\Cyl(f)/(X\times\{1\})$ of $\Cyl(f)$ with the subspace
$X\times\{1\}$ collapsed into a point.

By the discussion concerning attachments, if $X$ and $Y$ are cell complexes and $f$ is cellular then
$\Cyl(f)$ and $\Cone(f)$ are cell complexes as well. Moreover, if $f$ is a simplicial map between simplicial sets,
then by taking the analogous simplicial attachments and quotients, we obtain simplicial sets, denoted
by $\Cyl(f)$ and $\Cone(f)$ as well, and called the \emph{simplicial
mapping cylinder} and \emph{simplicial mapping cone}, respectively. The simplicial constructions are compatible
with geometric realizations; i.e., for example, $|\Cyl(f)|\cong\Cyl(|f|)$.
\begin{center}
\includegraphics{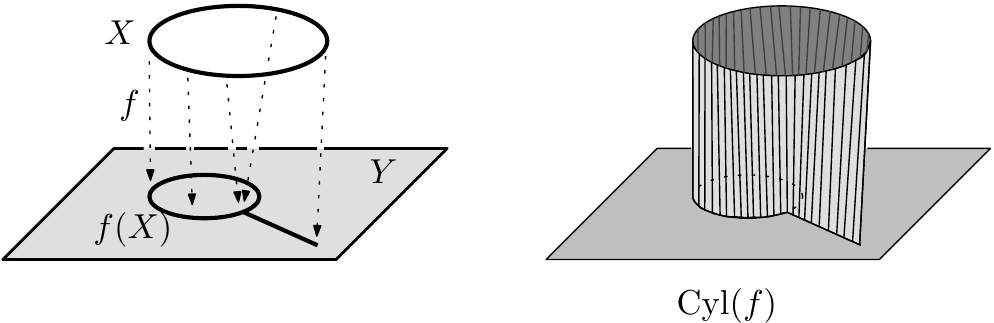}
\end{center}

We will use the mapping cylinder in our construction to replace an
arbitrary map $f\:X\to Y$ by an inclusion $X\hookrightarrow \Cyl(f)$,
 which has the same homotopy properties as $f$.
A more precise statement is given in the following lemma (see,
e.g., \cite[Corollary~0.21]{Hatcher}).

\begin{lemma}
\label{lem:mapping-cylinder} Let $f\:X\to Y$ be a continuous map between topological spaces. We consider $X\cong X\times\{1\}$ and $Y$ as subspaces
of $\Cyl(f)$ and denote the corresponding inclusion maps\footnote{More precisely, the inclusion maps are given as the composition of the respective inclusions $X\cong X\times\{1\} \subseteq X\times [0,1]\sqcup Y$ and
$Y\subseteq X\times [0,1]\sqcup Y$ with the quotient map $X\times [0,1]\sqcup Y \to \Cyl(f)$.}
by $i_X\:X\hookrightarrow \Cyl(f)$ and $i_Y\:Y\hookrightarrow \Cyl(f)$.
\begin{enumerate}
\item[\textup{(a)}] $Y$ is a strong deformation retract%
\footnote{We recall that a deformation retraction of a space $X$ onto a subspace $A$ is a map $H\:X\times[0,1]\to X$ such that $H(x,0)=x$ and $H(x,1)\in A$ for all $x\in X$ and $H(a,1)=a$ for all $a\in A$. Thus, a deformation retraction witnesses that the inclusion map $i_A\:A\hookrightarrow X$ is a homotopy equivalence with a homotopy inverse $r=H(\cdot,1)\:X\to A$ that is a retraction, i.e., that restricts to the identity on $A$.

A deformation retraction is called \emph{strong} if it keeps $A$ fixed pointwise throughout, i.e., if $H(a,t)=a$ for all $a\in A$ and $t\in [0,1]$ (some authors include this directly in the definition of a deformation retraction).}
of $\Cyl(f)$.
\item[\textup{(b)}] $X$ (considered as a subspace via $i_X$) is a strong deformation retract of
$\Cyl(f)$ iff $f$ is a homotopy equivalence.
\item[\textup{(c)}] $i_X\sim i_Y\circ f$ are homotopic as maps $X\to\Cyl(f)$.
\item[\textup{(d)}]If $f\colon X\rightarrow Y$ is a homotopy equivalence and if $g\colon Y\rightarrow X$ is a homotopy inverse for $f$, then
 $i_X\circ g\sim i_Y$ as well.
\end{enumerate}
\end{lemma}

\heading{Reduced mapping cone and mapping cylinder.}
If $X$ and $Y$ and $f$ are pointed, with basepoints $x_0$ and $y_0$, it will be technically
convenient, particularly in Section~\ref{s:poly-constructions}, to consider the spaces $\rcyl(f)$ and $\rcone(f)$,
called the \emph{reduced mapping cylinder} and the \emph{reduced mapping cone}, respectively, that are obtained from $\Cyl(f)$ and $\Cone(f)$ by collapsing the segment $x_0\times [0,1]$ (whose lower end is identified with $y_0$) to a single point. We will  apply this construction only
to cellular or simplicial mapping cylinders and cones, in which case contracting the subcomplex $x_0\times [0,1]$ is a homotopy equivalence.

Moreover, if $f$ is a homotopy equivalence then we may assume that its homotopy inverse $g$ is pointed as well and that
the homotopies $f\circ g\simeq \id_Y$ and $g\circ X\simeq \id_X$ keep the basepoints fixed (see \cite[Corollary~0.19]{Hatcher}). It follows that Lemma~\ref{lem:mapping-cylinder} remains true if we take $C=\rcyl(f)$ as the \emph{reduced} mapping cylinder (the inclusions are given as those into $\Cyl(f)$, followed by the quotient map $\Cyl(f)\to \rcyl(f)$, which does not make any identifications within $X$ or within $Y$).

By the remarks concerning the size of simplicial products and attachments, the size of the (reduced or unreduced) simplicial mapping cylinder or cone is at most the size of $X$ plus the size of $Y$, times a factor depending only on $\dim X$.

\subsection{Subdivisions and simplicial approximation}
\label{s:simplicial-approx}

For simplicial complexes, there is the well-known notion of \emph{barycentric subdivision} (see, e.g., \cite[\S 15]{Munkres}).
An analogous notion of subdivision, called \emph{normal subdivision}, can also be defined for simplicial sets.
Informally speaking, the normal subdivision $\sd(X)$ of a simplicial set $X$ is defined by barycentrically subdividing each simplex of $X$ and then glueing these subdivided simplices together according to the identifications implied by the face and degeneracy operators of $X$.
We refer to \cite[Section~4.6]{FritschPiccinini:CellularStructures-1990} for the precise formal definition and just state the facts that we will need in what follows.

For the standard simplex $\Delta^p$, the nondegenerate $k$-simplices of $\sd(\Delta^p)$ correspond to
chains of proper inclusions of nondegenerate simplices (faces) of $\Delta^p$. It follows that $\sd(\Delta^p)$
has $(p+1)!$ nondegenerate $p$-simplices and, in general, at most $2^{p+1}(p+1)!$ nondegenerate simplices
of any dimension. Consequently, for any simplicial set $X$, the size of $\sd(X)$ is at most the size of $X$
times a factor that depends only on $\dim X$ and which can be bounded from above by $2^{\dim X+1}(\dim X+1)!$.
Moreover, if the dimension is bounded, $\sd(X)$ can be constructed in time polynomial in $\size(X)$.

If $f\:X\to Y$ is a simplicial map, then subdivision also induces a map $\sd(f)\:\sd(X)\to \sd(Y)$, and this is compatible with compositions; i.e., $\sd(f\circ g)=\sd(f)\circ \sd(g)$.

For each simplicial set $X$, there is a simplicial map $\lastvertex_X\:\sd(X)\to X$, called the \emph{last vertex map},%
\footnote{On the standard simplex $\Delta^n$, seen as a simplicial set, this map is defined by sending a chain
$(\sigma_0,\ldots,\sigma_k)$ (a $k$-simplex of $\sd(\Delta^n)$) to the simplex $[v_0,\ldots,v_k]$, where $v_i$ is the last vertex of the simplex $\sigma_i$ (recall that the vertices in each simplex are ordered).} which is a \emph{homotopy equivalence} that is compatible with simplicial maps $f\:X\to Y$, i.e., $f\circ \lastvertex_X=\lastvertex_Y\circ \sd(f)$.\footnote{In the language of category theory, $\sd$ is a functor and $\lastvertex$ is a natural transformation between $\sd$ and the identity functor on simplicial sets.} \footnote{In fact, it is true that $X$ and $\sd(X)$ are not only homotopy equivalent but homeomorphic (as one might expect given the terminology `` subdivision''). However, for simplicial sets this is a decidedly nontrivial result, see  \cite[Cor.~4.6.5]{FritschPiccinini:CellularStructures-1990}. The difficulty is related to the fact that there is no way of defining this homeomorphism for all simplicial sets in such a way that it becomes compatible with simplicial maps. For our purposes, the natural homotopy equivalence $\lastvertex_X$ will be sufficient and more convenient.}

There is also a simplicial approximation theorem for simplicial sets, which uses iterated normal subdivisions.
Specifically, the $\sddim$-fold iterated normal subdivision of a simplicial set is defined inductively as $\sdit(X):=\sd(\sd^{\sddim-1}(X))$,
where $\sd^0(X):=X$.

\begin{theorem}[{\cite[Thm.~4.~6.~25]{FritschPiccinini:CellularStructures-1990}}]
\label{t:s-apx-sset}
Let $X$ and $Y$ be simplicial sets such that $X$ has only finitely many nondegenerate simplices,
and let $f\:|X|\rightarrow |Y|$ be a continuous map. Then there exist a finite integer $\sddim$ (which depends on $f$) and a
simplicial map $g\:\sdit(X)\to Y$ such that  $|g|$ is homotopic to the composition $f\circ |\lastvertex_X^t|$
of $f$ with the iterated last vertex map $\lastvertex_X^t\:\sdit(X)\to Y$.
\end{theorem}

To convert arbitrary simplicial sets into homotopy-equivalent (in fact, homeomorphic) simplicial complexes, another subdivision-like operation is needed, (see, e.g., \cite{Jardine:SimplicialApproximation-2004}). Given a simplicial set $Z$, one can define a simplicial complex $B_\ast(Z)$ inductively,  by
introducing a new vertex $v_\sigma$ for every \emph{nondegenerate} simplex $\sigma$, and then replacing $\sigma$ by the cone with apex $v_\sigma$ over $B_*(\partial \sigma)$. If the simplicial set $Z$ has a certain regularity property---which is satisfied, for instance, if $Z=\sd(X)$---then $B_\ast(Z)$ and $Z$ are homotopy equivalent (in fact, homeomorphic).\footnote{As an illustration that this fails for general simplicial sets, consider the case where $Z=\Sigma^p$ is the simplicial set model of the $d$-sphere with only two nondegenerate simplices, one in dimension $0$ and one in dimension $d$. In this case, $B_\ast (\Sigma^d)$ is a $1$-dimensional simplex.}  We summarize the properties that we need in the following proposition (for completeness, we provide a proof in \ref{a:subdivisions}).

\begin{prop}
\label{p:doublesubdiv}
If $X$ is a simplicial set,
 then the \emph{twofold subdivision} $B_*(\sd(X))$ is
a simplicial complex.
Moreover, there is a simplicial map $\gamma_X\:B_*(\sd(X))\to X$,
which is a homotopy equivalence.
For a simplicial subset $A\subseteq X$, $B_*(\sd(A))$ is a subcomplex of
$B_*(\sd(X))$ and $\gamma_X|_A=\gamma_A$.\footnote{In fact, the construction
$B_*\sd$ is functorial and $\gamma$ is a natural transformation
(like the construction $\sd$ and the map $\lastvertex$),
but we will never use this stronger fact.}

If $X$ is finite and of bounded dimension, there are algorithms that construct the simplicial
complex $B_*(\sd(X))$ and evaluate the map $\gamma_X$, both in polynomial time.
\end{prop}

\newdir{ >}{{}*!/-5pt/@{>}}

\section{Homotopy groups}
\label{s:homotopy-groups}

We review some further facts about homotopy groups that we will
need. For more details see, e.g., \cite[Section
4.1]{Hatcher}.

\subsection{Basic facts}
So far, we used the definition of the $\then$th homotopy
group $\pi_\then(X,x_0)$ of a pointed space $(X,x_0)$ as the set of homotopy classes
of pointed maps $(S^\then,p_0)\to (X,x_0)$, where $p_0\in S^\then$ is an
arbitrarily chosen basepoint, and the homotopies are required to
keep the basepoint fixed. Equivalently,
the elements of $\pi_\then(X,x_0)$
can be viewed as homotopy classes $[f]$ of maps $f\:(D^\then,\partial D^\then)\to (X,x_0)$ sending
all of $\partial D^\then$ to $x_0$, modulo homotopies that keep the image of $\partial D^\then$ fixed (as before, we will often drop the basepoint from the notation).\footnote{The claimed equivalence is obtained by identifying
$S^\then$ with
the quotient $D^\then/\partial D^\then$ of the $\then$-disk by its boundary and $p_0$ with the
image of $\partial D^\then$ under the quotient map.}

In what follows, we will also need the Abelian group operation in $\pi_\then(X,x_0)$, $\then\geq 2$, which can be defined
as follows: Suppose $f_1,\ldots, f_\thelen$ are maps $(D^\then,\partial D^\then)\to (X,x_0)$. Suppose we have
a cellular decomposition of $D^\then$ as a cell complex $\chsimp\then\thelen$ with
$\then$-cells $e_1^\then,\ldots,e_\thelen^\then$
(in Section~\ref{s:sum}
 below we will provide a concrete geometric construction of
$\chsimp\then\thelen$).
Then we can define a map $f$ from $D^\then\cong\chsimp\then\thelen$ to $(X,x_0)$ representing the homotopy class $[f_1]+\ldots +[f_\thelen]$
by sending the $(\then-1)$-skeleton of $\chsimp\then\thelen$ to $x_0$, and by defining the restriction of $f$ to each open cell $e_i^\then$ to be $f_i$.%
%

A very important special case of homotopy groups
are those of spheres. We will use
the following well-known facts:\begin{itemize}
\item
The sphere $S^\then$ is $(\then-1)$-connected.
\item
For all $\then\ge 1$, $\pi_\then(S^\then)$ is isomorphic to $\Z$
and generated by the homotopy class $\iota$ of the identity.
For each map $\varphi\:S^\then\to S^\then$, there is a unique
integer $a\in\Z$ such that $[\varphi]=a\iota$; it is
called the \indef{degree} of $\varphi$ and denoted by $\deg\varphi$.
The degree is obviously invariant under homotopy.
\item
We have $\pi_3(S^2)\cong \Z$. The group is generated by the
famous \emph{Hopf map}\footnote{See \cite[Ex.~4.45]{Hatcher} for the definition (which is not difficult, but which we will not need).} $\eta\:S^3\to S^2$.
\end{itemize}

We will also need the following simple fact:
\begin{lemma}
\label{lem:scalar-multiplication-degree} Let $g\colon
S^\then\rightarrow S^\then$ be a map of degree $b\in \Z$. Then, for
any map $f\colon S^\then\rightarrow X$, we have $[f\circ g]=b\cdot [f]\in \pi_\then(X)$.
\end{lemma}
\begin{proof}
Consider the $\then$-dimensional unit cube $I^\then\cong D^\then$,
where $I=[0,1]$ is the unit interval. We identify $S^\then$ with the quotient
$I^\then/\partial I^\then$. From the map of sets $(I^\then,\partial I^\then) \rightarrow (I^\then,\partial I^\then)$
given by $(s_1,\ldots,s_{\then-1},s_\then)\mapsto (s_1,\ldots,s_{\then-1}, b s_\then \bmod 1)$ we obtain $g_0\:S^\then\to S^\then$ by passing to quotients. By
the definition of the addition of homotopy classes, on the one hand, $[g_0]$ is
the $b$-fold sum of the identity, and hence a particular
example of a map of degree $b$.
On the other hand, $f\circ g_0$ is a representative of $b\cdot [f]$, the $b$-fold sum of $[f]$.
Since $[f\circ g]$ depends only on the homotopy class $[g]$, which is uniquely determined
by the degree of $g$, the lemma follows.
\end{proof}

Let $X$ be a cell complex and $A\subseteq X$ a subcomplex.
Then the homotopy groups of the spaces
$A$, $X$ and $X/A$ in a certain range are connected by an exact sequence.

\begin{theorem}\label{t:exact} Let $A\subseteq X$ be cell complexes. Let $p$, $q\ge 0$ be integers, $q\leq p+1$. If $A$ is $p$-connected and $X$ is $q$ connected, then
$X/A$ is also $q$-connected and there is an exact sequence
$$\pi_{p+q}(A)\to\dots\to\pi_i(A) \to \pi_i(X) \to \pi_i(X/A)\to \pi_{i-1}(A)\to\dots\to
\pi_{q+1}(X/A).$$
Here the maps $\pi_i(A)\to\pi_i(X)$ and $\pi_i(X)\to\pi_i(X/A)$ are induced by the inclusion and the projection, respectively, and the exactness
means that the kernel of each homomorphism equals the image of the preceding one.
\end{theorem}

\begin{proof} One can define homotopy groups of any pair $(X,A)$, $A\subseteq X$, and these homotopy groups fit into the following exact sequence
\[
\dots\to\pi_i(A) \to \pi_i(X) \to \pi_i(X,A)\to \pi_{i-1}(A)\to\dots, \quad i\ge 1;\]
see \cite{Hatcher}, Section~4.1. From the exactness and the connectivity assumptions
it is easy to show that $\pi_i(X,A)=0$ for $i\leq q$.  Then,
 according to
\cite[Proposition~4.28]{Hatcher},
 the map $\pi_i(X,A)\to\pi_i(X/A)$ induced by the quotient map $X\to X/A$
is an isomorphism for $i\leq p+q$. Substituting  $\pi_i(X/A)$ in this range into the exact sequence above, we get the exact sequence from the statement of the theorem.
\end{proof}

In the proof of Theorem~\ref{t:undecide} we will need a description of the $\then$th homotopy group of a cell complex $Y$ obtained from $T$ by attaching  $(\then+1)$-cells $e_{\ell}$, $1\leq \ell\leq m$,
by attaching maps $\varphi_{\ell}\:S^{\then}\to T$.

\begin{prop}\label{p:attach}
Let $\then\ge 2$ be an integer. Suppose that $T$ is a
$1$-connected cell complex and $Y$ is a cell complex obtained from $T$
as described above. Then
$$\pi_{\then}(Y)\cong\pi_{\then}(T)/\langle [\varphi_1], [\varphi_2],\dots,[\varphi_m]\rangle,
$$
where $\langle [\varphi_1], [\varphi_2],\dots,[\varphi_m]\rangle$ is the subgroup of $\pi_{\then}(T)$ generated by the homotopy classes of $\varphi_{\ell}$, $1\leq \ell\leq m$.
\end{prop}

\begin{proof} It is sufficient to prove the statement for a single cell attached; then
we can proceed by induction. Consider the reduced mapping cylinder $\rcyl{\varphi}$,
together with the inclusions of $S^{\then}$ and $T$ into it and the projection onto $\rcyl{\varphi}/S^\then=
\rcone{\varphi}=Y$. The situation is summarized in the diagram
\[
\xymatrix{
S^{\then} \ar[rd]_{\varphi} \ar[r] & \rcyl{\varphi} \ar[r]  & \rcyl{\varphi}/S^{\then}=Y \\
& T \ar[u]_{\sim} &
}\]
which commutes up to homotopy. Applying Theorem~\ref{t:exact} for $A=S^{\then}$,
$X=\rcyl{\varphi}$, $p=\then-1$ and $q=1$, we obtain the exact sequence
$$\pi_{\then}(S^{\then})\to \pi_{\then}(\rcyl{\varphi})\to \pi_{\then}(Y)\to\pi_{\then-1}(S^{\then})=0.$$
If we replace the inclusion $S^{\then}\hookrightarrow\rcyl{\varphi}$ by the map
$\varphi\:S^{\then}\to T$, we get
$$\pi_{\then}(S^{\then})\buildrel{\varphi_*}\over\longrightarrow \pi_{\then}(T)\to \pi_{\then}(Y)\to 0.$$
Hence $\pi_{\then}(Y)=\pi_{\then}(T)/\langle\varphi_*(\iota)\rangle$,
 where $\iota$ is the homotopy class of the identity on $S^{\then}$,
and thus  $\varphi_*(\iota)$ is the homotopy class of $\varphi$.
\end{proof}

\subsection{Whitehead products and wedge sums of spheres}\label{s:Whitehead}
\heading{Whitehead products.} There is another type of operation on elements of homotopy groups that we will need.
Consider two spheres $S^k$ and $S^\ell$ with their standard structures as cell complexes (one vertex and one cell of
the top dimension). Then the product $S^k\times S^\ell$ is also a cell complex, with one vertex, one respective cell $e^k$ and $e^\ell$
in dimensions $k$ and $\ell$, and one cell $e^k\times e^\ell$ in dimension $k+\ell$. In particular, the $(k+\ell-1)$-skeleton of the product is a wedge $S^k\vee S^\ell$, to which the $(k+\ell)$-cell is attached via a map $\varphi\:S^{k+\ell-1}\cong \partial (D^{k+\ell})\to S^k \vee S^\ell$.

Now, if $f\:S^k\to X$ and $g\:S^\ell \to X$ are (pointed) maps, we can combine them to a map $f\vee g\: S^k\vee S^\ell\to X$.
If we compose this with the attachment map $\varphi$ discussed before, we get a map $[f,g]\:S^{k+\ell-1}\to X$, called the
\emph{Whitehead product} of $f$ and $g$. The homotopy class of this product clearly depends only on the homotopy classes of the factors, so we get a well-defined product $\pi_k(X)\times \pi_\ell(X)\to \pi_{k+\ell-1}(X)$,
again  denoted by $[\cdot,\cdot]$.
As a quite trivial but nonetheless useful example, if $X=S^k\times S^\ell$, then the attachment map $\varphi$ itself equals the Whitehead product $[\iota_{S^k},\iota_{S^\ell}]$ of the two inclusions $\iota_{S^k}\:S^k\hookrightarrow S^k\vee S^\ell$ and
$\iota_{S^\ell}\:S^\ell \hookrightarrow S^k\vee S^\ell$.

In our proofs we will use that the Whitehead product is natural, graded commutative
and bilinear, i.e.
\begin{align*}
f_*[\alpha,\beta]&=[f_*\alpha,f_*\beta],\\
[\alpha,\beta]&=(-1)^{k\ell}[\beta,\alpha],\\
[\alpha+\gamma,\beta]&=[\alpha,\beta]+[\gamma,\beta],\\
[\alpha,\beta+\delta]&=[\alpha,\beta]+[\alpha,\delta]
\end{align*}
where $\alpha,\ \gamma\in\pi_k(X)$, $\beta,\ \delta\in\pi_{\ell}(X)$ and $f\:X\to Y$.
For the proof see \cite{Whitehead:HomotopyTheory-1978}, Chapter X, 7.2, Cor.~7.12 and Cor.~8.13.

In the proof of Theorem~\ref{t:undecide} we will need some facts about the homotopy groups
of spheres and their wedge sums.

\begin{theorem}[{\cite[Cor.~4B.2]{Hatcher},\cite[XI, Thm.~2.5]{Whitehead:HomotopyTheory-1978}}]
\label{t:hopf}
There is a homomorphism  (called the Hopf invariant)
$H\:\pi_{2\thedim-1}(S^\thedim)\to \mathbb Z$
such that for $d$ even $H([\iota,\iota])=\pm 2$.
\end{theorem}

Let us note that for $\thedim$ odd the Whitehead product $[\iota,\iota]\in\pi_{2\thedim-1}(S^\thedim)$ is of order two, i.e.~$2[\iota,\iota]=0$.
Whitehead products play a crucial role in Hilton's theorem which converts the computation of homotopy groups of a wedge of spheres to the computations of homotopy groups of spheres.
We do not need this theorem in its full generality as it was proved in \cite{Hilton}, and so we restrict ourselves to a special case.

Let $\thedim\ge 2$ and $r$, $s\ge 1$ be integers. Let
\begin{equation}\label{e:T}
T= S_1^\thedim\vee \cdots \vee S_r^\thedim\vee  S^{2\thedim-1}_1\vee\cdots \vee S_s^{2\thedim-1}
\end{equation}
be the wedge sums of $r$ copies of $S^\thedim$ and $s$ copies of $S^{2\thedim-1}$.
Denote by $\nu_i$ and $\mu_\theeqn$ the homotopy classes of the inclusions $S^\thedim
\hookrightarrow T$ and $S^{2\thedim-1}\hookrightarrow T$ onto the $i$th copy of $S^\thedim$ and the $\theeqn$th copy of $S^{2\thedim-1}$, respectively. Then the homotopy groups $\pi_\thedim(T)$ and $\pi_{2\thedim-1}(T)$ can be described by the following special case of Hilton's theorem.

\begin{theorem}[{\cite[Thm.~A]{Hilton}}]
\label{t:hilton}
With the notation as above, there are isomorphisms
\begin{align*}
\pi_\thedim(T)&\cong \bigoplus_{1\leq i\leq r}\pi_\thedim(S_i^\thedim),\\
\pi_{2\thedim-1}(T)&\cong\bigoplus_{1\leq i\leq r} \pi_{2\thedim-1}(S_i^\thedim)\oplus \bigoplus_{1\leq \theeqn \leq s}\pi_{2\thedim-1}(S_\theeqn^{2\thedim-1})\oplus \bigoplus_{1\leq i<j\leq r}\pi_{2\thedim-1}(S_{ij}^{2\thedim-1}).
\end{align*}
An element $\beta\in \pi_{2\thedim-1}(S_i^\thedim)$ corresponds to the composition $\nu_i\circ\beta\in \pi_{2\thedim-1}(T)$,
an element $\beta\in \pi_{2\thedim-1}(S_\theeqn^{2d-1})$ to the composition
$\mu_\theeqn\circ\beta \in \pi_{2\thedim-1}(T)$, and
an element $\beta\in \pi_{2\thedim-1}(S_{ij}^{2\thedim-1})$
to the composition $[\nu_i,\nu_j]\circ\beta\in \pi_{2\thedim-1}(T)$.
\end{theorem}

We will say that some elements $x_1,\ldots,x_r$ of an Abelian group are
\emph{integrally independent}
 if the only valid relation $a_1x_1+\cdots+a_rx_r=0$ with coefficients $a_i\in\Z$ is that with all $a_i$ zero. The following statement is an immediate consequence of Theorems~\ref{t:hilton} and \ref{t:hopf}.

\begin{corol}\label{c:freegen} If $\thedim\ge 2$ is odd, then the elements $\mu_\theeqn$,
$1\leq \theeqn \leq s$, and $[\nu_i,\nu_j]$, $1\leq i<j\leq r$ are integrally independent in $\pi_{2\thedim-1}(T)$.

If $\thedim\ge 2$ is even, then the elements $\mu_\theeqn$, $1\leq \theeqn\leq s$ and
$[\nu_i,\nu_j]$, $1\leq i\leq j\leq r$ are integrally independent in $\pi_{2\thedim-1}(T)$.
\end{corol}

\begin{proof}
The reason is that every element in the list comes from a different direct summand and is of infinite order.
\end{proof}

In the case $\thedim=2$ and $s=0$ we can say even more:

\begin{corol}[{\cite[Ex.~4.52]{Hatcher}}]
\label{c:twospheres}
The homotopy group
$\pi_{3}(\bigvee_{i=1}^{r}S_i^2)$ is a free Abelian group generated by
the Whitehead products $[\nu_i,\nu_j]$, $1\leq i<j\leq r$, and homotopy classes $\nu_i\eta$, where $\eta\:S^3\to S^2$ is the Hopf map.
\end{corol}

\section{The constructions for Theorem~\ref{t:undecide} presented
as cell complexes}
\label{s:constructions}

Here we present the essence of the proof of Theorem~\ref{t:undecide}.
Namely, for every system of quadratic Diophantine equations
of the form \eqref{e:qsym} (for $\thedim$ even) or \eqref{e:qskew}
(for $\thedim$ odd), we construct cell complexes $A$, $X$, $Y$,
and a continuous map $f\:A\to Y$, where $Y$ is
$(\thedim-1)$-connected and $\dim X=2\thedim$, such that $f$ is extendable
to $X$ iff the Diophantine system has a solution. Moreover,
one of $(X,A)$ and $Y$ can be assumed to be fixed, as in
Theorem~\ref{t:undecide}~(a) and~(b).
We will also see  the role of Whitehead products and Hilton's theorem in
the proof.

What remains for the next section is to convert $X$, $A$, $Y$ into finite simplicial
complexes and $f$ into a simplicial map, so that the solvability of
the extension problem remains unchanged. Moreover, the construction
has to be algorithmic.

While discussing the cellular constructions of $X$, $A$, $Y$, it is also
natural to describe the cell complex used by Anick in the proof of
his \#P-hardness result. Indeed, his construction uses tools very
similar to those employed in our constructions.



\heading{The generalized extension problem.}
In order to simplify the presentation, it is convenient to
remove the assumption in the extension problem
that $A$ is a subspace of $X$, or
in other words, that the map $A\to X$ is an inclusion.
Instead, we consider three spaces $A$, $W$, $Y$ and (arbitrary)
maps $g\:A\to W$ and $f\:A\to Y$, and we ask if there is a map
$h\:W\to Y$ making the following diagram commutative up to homotopy:
{\renewcommand\theequation{GEP}
\begin{equation}\label{e:triangle}
\xymatrix{
A \ar[r]^-f \ar[d]_-g & Y \\
W \ar@{-->}[ur]_-h &
}
\end{equation}}%
For a generalized extension problem as above,
we obtain an equivalent extension problem by
setting $X=\Cyl(g)$ (where $A$ is considered as a subspace
of the cylinder in the usual way). This is easy to see, but we
nonetheless briefly describe a proof of this fact that only uses
the homotopy extension property for pairs of cell complexes and the
properties of the mapping cylinder summarized in Lemma~\ref{lem:mapping-cylinder}.
This means that we can replace the mapping cylinder $\Cyl(g)$ by any other cell complex
that has these properties, and the same proof will still apply.
This will be useful for our simplicial constructions later on, for which it will be convenient
to work with so-called \emph{generalized mapping cylinders} (see Section~\ref{s:gmc}).

Let $i_A$ and $i_W$ be the inclusions of $A$ and $W$ into $X$.
On the one hand, given a solution $\bar{f}\:X\to Y$ of the extension problem,
i.e., $\bar{f}\circ i_A=f$, we can define $h:=\bar{f}\circ i_W$ as the restriction of
$\bar{f}$ to $W$. Then $h\circ g=\bar{f}\circ i_W \circ g\sim \bar{f}\circ i_A=f$,
so $h$ is a solution to the generalized extension problem.

On the other hand, given a solution $h$ for the generalized extension problem
(\textup{GEP}), let $r\:X\to W$ be the retraction from $X$ to $W$ and
define $\bar{f}:=h\circ r$. Then $\bar{f}\circ i_A=h\circ r\circ i_A\sim h\circ g\sim f$,
so $\bar{f}$ is an extension of a map homotopic to $f$, and since extendability depends
only on the homotopy class of a map (Corollary~\ref{c:ext}), $f$ can be extended as well.

Thus, we are free to consider the
generalized extension problem with $\dim W\le 2\thedim$ and
$\dim A\le 2\thedim-1$.

\subsection{Fixed target}\label{s:fixedtarget}
We describe an instance of the generalized extension problem
for part (a) of Theorem~\ref{t:undecide}, where the target $Y$ is fixed.
The system of equations will be encoded into cell complexes $A$, $W$ and the maps $g\:A \to W$, $f\:A\to Y$.

\heading{Fixed target with $\thedim$ even.} Here $Y=S^\thedim$,
\begin{equation}\label{e:AW}
A=S^{2\thedim-1}_1\vee\cdots \vee S_s^{2\thedim-1},\quad W=S_1^\thedim\vee \cdots \vee S_r^\thedim.
\end{equation}
Then the diagram~\eqref{e:triangle} becomes
\[\xymatrix{
\rightbox{A={}}{S^{2\thedim-1}_1\vee\cdots\vee S^{2\thedim-1}_s} \ar[r]^-f \ar[d]_g & S^\thedim \\
\rightbox{W={}}{S^\thedim_1\vee\cdots\vee S^\thedim_r} \ar@{-->}[ur]_-h &
}\]
According to \eqref{e:wedge} from the discussion of wedge sums in Section~\ref{s:general-constructions},
the homotopy class of $f\:A\to S^\thedim$ is specified completely by the homotopy classes of its restrictions
to the spheres forming $A$. We will use $f_*\mu_\theeqn=[f]\mu_\theeqn\in\pi_{2\thedim-1}(S^\thedim)$ to denote these, where $\mu_\theeqn$ is the homotopy class of the inclusion of the $\theeqn$th sphere $S^{2\thedim-1}_\theeqn$ into $A$. Our particular choice is
\begin{equation}\label{e:f-fixed-target-k-even}
f_*\mu_\theeqn=b_\theeqn[\iota,\iota].
\end{equation}
Similarly, the homotopy class of $g\:A\to W$ is given by its restrictions as
\begin{equation}\label{e:g-fixed-target-k-even}
g_*\mu_\theeqn=\sum_{1\leq i<j\leq r}\ a_{ij}^{(\theeqn)}[\nu_i,\nu_j],
\end{equation}
where $\nu_i$ the homotopy class of the inclusion of the $i$th spere $S^\thedim_i$ into $W$. Finally, let $h\:W\to S^\thedim$ be an arbitrary map and write $h_*\nu_i=x_i\iota\in\pi_\thedim(S^\thedim)$ for some integers $x_i\in\Z$. According to \eqref{e:wedge} again, the diagram
\eqref{e:triangle} commutes up to homotopy iff $(hg)_*\mu_\theeqn=f_*\mu_\theeqn$, i.e.~iff
\begin{equation}\label{e:Sd}
h_*\Bigl(\sum_{1\leq i<j\leq r}\ a_{ij}^{(\theeqn)}[\nu_i,\nu_j]\Big)=b_\theeqn[\iota,\iota]
\end{equation}
for all $1\leq\theeqn\leq s$. Using the naturality and bilinearity of the Whitehead product,
the left hand side equals
$$\sum_{1\leq i<j\leq r}\ a_{ij}^{(\theeqn)}[h_*\nu_i,h_*\nu_j]=
\sum_{1\leq i<j\leq r}\ a_{ij}^{(\theeqn)}[x_i\iota,x_j\iota]=
\sum_{1\leq i<j\leq r}\ a_{ij}^{(\theeqn)}x_ix_j[\iota,\iota].$$
According to Theorem~\ref{c:freegen} the homotopy class $[\iota,\iota]$ is of infinite order, and so
the system of equations \eqref{e:Sd} is equivalent to \eqref{e:qsym}.
We get the following:
\begin{prop}\label{p:fixed-target-k-even}
Let the maps $f\:A\to S^k$ and $g\:A\to W$ be as in~\eqref{e:f-fixed-target-k-even} and \eqref{e:g-fixed-target-k-even} above. Then $f$ can be extended to $X=\rcyl(g)$ if and only if the system \eqref{e:qsym} has a solution.
\end{prop}

\heading{Fixed target with $\thedim$ odd.}  The element $[\iota,\iota]\in\pi_{2\thedim-1}(S^\thedim)$ has order $2$, so we cannot use $Y=S^\thedim$. However, leaving $A$, $W$ and $g\:A\to W$ as before, we can take  $Y=S^\thedim\vee S^\thedim$ and specify $f$ by
\begin{equation}\label{e:f-fixed-target-k-odd}
f_*\mu_\theeqn=b_\theeqn[\iota_1,\iota_2],\end{equation}
where $\iota_1$ and $\iota_2$ are inclusions of $S^\thedim$ onto the first and the second summand in $S^\thedim\vee S^\thedim$, respectively. Using
Hilton's theorem (Theorem~\ref{t:hilton})
for $\pi_\thedim(S^\thedim\vee S^\thedim)$, the homotopy class of a general map $h\:W\to Y$ satisfies
\begin{equation}\label{e:g-fixed-target-k-odd}
h_*\circ\nu_i=x_i\iota_1+y_i\iota_2,\quad x_i,y_i\in\mathbb Z.
\end{equation}
Using the fact that $[\iota_1,\iota_2]=-[\iota_2,\iota_1]$, it is easy to show that the commutativity of the diagram \eqref{e:triangle} is equivalent to the system of $s$ equations in $\pi_{2\thedim-1}(S^\thedim\vee S^\thedim)$,
$$\Big(\sum_{i<j}\ a_{ij}^{(\theeqn)}(x_iy_j-x_jy_i)\Big)[\iota_1,\iota_2]+
\Big(\sum_{i,j}\ a_{ij}^{(\theeqn)}x_ix_j\Big)[\iota_1,\iota_1]+
\Big(\sum_{i,j}\ a_{ij}^{(\theeqn)}y_iy_j\Big)[\iota_2,\iota_2]
=b_\theeqn[\iota_1,\iota_2].$$
By Corollary~\ref{c:freegen} of Hilton's theorem the element $[\iota_1,\iota_2]\in\pi_{2\thedim-1}(S^d\vee S^d)$ is of infinite order, while $[\iota_1,\iota_1]$ and $[\iota_2,\iota_2]$ are of order $2$. Multiplying all the equations in \eqref{e:qskew} by $2$, we get an equivalent system, in which all the $a_{ij}^{(\theeqn)}$ are even. For this system, the above equation is exactly the one from \eqref{e:qskew}.
We get the following:
\begin{prop}\label{p:fixed-target-k-odd}
Let the maps $f\:A\to S^k\vee S^k$ and $g\:A\to W$ be as in~\eqref{e:f-fixed-target-k-odd} and \eqref{e:g-fixed-target-k-even} above. Then $f$ can be extended to $X=\rcyl(g)$ if and only if the system \eqref{e:qskew} has a solution.
\end{prop}

\subsection{Fixed source} \label{s:fixedsource}
The idea for the first step of the proof of Theorem~\ref{t:undecide}~(b) is similar, only the constructions involve attaching cells and also the usage of Hilton's theorem is more substantial.

\heading{Fixed source with $\thedim$ even.} We put
$$A=S^{2\thedim-1},\quad W=S^\thedim$$
where the homotopy class of $g\:A\to W$ is $[g]=[\iota,\iota]$.
A given system of equations \eqref{e:qsym} will be encoded in the target space $Y$ and in the homotopy class of $f\:A\to Y$. The target space $Y$ is
a cell complex obtained from the wedge of spheres $T$ defined in \eqref{e:T} by attaching $(2\thedim)$-cells $e_{ij}$, $1\leq i<j\leq r$ and $e_{ii}$, $1\leq i\leq r$, i.e.
\begin{equation}\label{e:target}
Y=\underbrace{(S_1^\thedim\vee \cdots \vee S_r^\thedim\vee  S^{2\thedim-1}_1\vee\cdots \vee S_s^{2\thedim-1})}_T\cup\bigcup_{1\leq i<j\leq r}e_{ij}\cup\bigcup_{1\leq i\leq r}e_{ii}.
\end{equation}
The attaching maps for the cells are the maps $S^{2\thedim-1}\to T$ whose homotopy classes are, respectively,
\[\varphi_{ij}=[\nu_i,\nu_j]-\sum_{1\leq\theeqn\leq s}a_{ij}^{(\theeqn)}\mu_\theeqn,\ \ \
\varphi_{ii}=[\nu_i,\nu_i].
\]
Denote the images of the homotopy classes $\mu_\theeqn\in \pi_{2\thedim-1}(T)$, $1\leq\theeqn\leq s$ and $\nu_i\in\pi_\thedim(T)$, $1\leq i\leq r$,
 in $Y$ by $\mu'_\theeqn$ and $\nu'_i$, respectively.
Further, take a map $f\:A\to Y$ of the homotopy class
$$[f]=2b_1\mu'_1+2b_2\mu'_2+\dots+2b_s\mu'_s.$$
Since $\pi_\thedim(Y)\cong\pi_\thedim(T)\cong\pi_\thedim(S_1^\thedim)\oplus\cdots\oplus\pi_\thedim(S_r^\thedim)$ by Theorem~\ref{t:hilton}, a general map
$h\:W\to Y$ has a homotopy class
$$[h]=x_1\nu'_1+x_2\nu'_2+\dots+x_r\nu'_r$$
with arbitrary integer coefficients $x_i$.
To show that the commutativity of the diagram \eqref{e:triangle}
 (up to homotopy)
is equivalent to
the satisfaction of the system \eqref{e:qsym},
 we will need the following lemma.

\begin{lemma}\label{l:piY}
Let $Y$ be the cell complex as above. Then the classes $\mu'_\theeqn\in \pi_{2\thedim-1}(Y)$,
$1\leq\theeqn\leq s$, are integrally independent and
\begin{align*}
[\nu'_i,\nu'_j]&=\sum_{1\leq\theeqn\leq s} a_{ij}^{(\theeqn)}\mu'_\theeqn,  \quad 1\leq i<j\leq r,\\
[\nu'_i,\nu'_i]&=0 , \qquad 1\leq i\leq r.
\end{align*}
\end{lemma}

\begin{proof}  The statement is a consequence of Proposition~\ref{p:attach}
and Corollary~\ref{c:freegen}.
\end{proof}

Using this lemma and the bilinearity and graded commutativity of the Whitehead product,
we compute $[hg]\in\pi_{2\thedim-1}(Y)$ as
\begin{align*}
h_*[g] & =h_*[\iota,\iota]=[h_*\iota,h_*\iota] \\
& =\bigl{[}\sum_{1\leq i\leq r} x_i\nu'_i,\sum_{1\leq j\leq r} x_j\nu'_j\bigr{]}=\sum_{1\leq i,j\leq r}x_ix_j[\nu'_i,\nu'_j] \\
& =2\sum_{1\leq i<j\leq r}x_ix_j[\nu'_i,\nu'_j] +\sum_{1\leq i\leq r} x_i^2[\nu'_i,\nu'_i] \\
& =2\sum_{1\leq i<j\leq r}x_ix_j\bigl{(}\sum_{\theeqn=1}^s a_{ij}^{(\theeqn)}\mu'_\theeqn\bigr{)} \\
& =2\sum_{1\leq \theeqn\leq s} \bigl{(}\sum_{1\leq i<j\leq r}a_{ij}^{(\theeqn)}x_ix_j \bigr{)}\mu'_\theeqn
\end{align*}
Comparing with $[f]$ and using the fact that $\mu'_\theeqn$ are integrally independent, we obtain the system \eqref{e:qsym}.

\heading{Fixed source with odd $\thedim$.}
As in the fixed target case, we resolve the problem of $[\iota,\iota]$ being of order $2$ by replacing it with $[\iota_1,\iota_2].$ In this case, it means that we set $A=S^{2\thedim-1}$, $W=S^\thedim\vee S^\thedim$. The target space $Y$ remains the same as for $\thedim$ even. We take $f$ to be any map with
$$[f]=b_1\mu'_1+b_2\mu'_2+\dots+b_s\mu'_s$$
and $g$ has the advertised homotopy class $[g]=[\iota_1,\iota_2]$, where $\iota_1$ and $\iota_2$ are the homotopy classes of the inclusions of the two copies of $S^\thedim$ into $W=S^\thedim\vee S^\thedim$. The homotopy class of a map
$h\:W\to Y$ is again determined by its restrictions along $\iota_1$, $\iota_2$, namely
$$h_*\iota_1=x_1\nu'_1+x_2\nu'_2+\dots+x_r\nu'_r, \quad
h_*\iota_2=y_1\nu'_1+y_2\nu'_2+\dots+y_r\nu'_r,$$
where the
$x_i$ and $y_i$ can be arbitrary integers.
The composition $[hg]\in\pi_{2\thedim-1}(Y)$ equals
\begin{align*}
h_*[g]=&\bigl{[}\sum_{1\leq i\leq r} x_i\nu'_i,\sum_{1\leq j\leq r} y_j\nu'_j\bigr{]}=\sum_{1\leq i,j\leq r}
x_iy_j[\nu'_i,\nu'_j]\\
&=\sum_{1\leq i<j\leq r}(x_iy_j - x_jy_i)[\nu'_i,\nu'_j]
+\sum_{1\leq i\leq r} x_i y_i[\nu'_i,\nu'_i]\\
&=\sum_{1\leq i<j\leq r}(x_iy_j-x_j y_i)\bigl{(}\sum_{1\leq\theeqn\leq s} a_{ij}^{(\theeqn)}\mu'_\theeqn\bigr{)}\\
&=\sum_{1\leq\theeqn\leq s} \bigl{(}\sum_{1\leq i<j\leq r}a_{ij}^{(\theeqn)}(x_iy_j-x_j x_i) \bigr{)}\mu'_\theeqn
\end{align*}
(again using Lemma~\ref{l:piY}). Since the
$\mu'_\theeqn$ are integrally independent, the comparison
with $[f]$ leads to the system~\eqref{e:qskew}.

Summarizing our findings, for $k$ both even and odd we get the following:
\begin{prop}\label{p:fixed-source}
For each $k\geq 2$ let the maps $f\:A\to Y$ and $g\:A\to W$ be as above.
Then $f$ can be extended to $X=\rcyl(g)$ if and only if there is a solution to the system \eqref{e:qsym} when $k$ is even, or \eqref{e:qskew} when $k$ is odd.
\end{prop}

\subsection{Anick's $4$-dimensional cell complexes}
\label{s:anick-complexes}

Here we introduce  complexes constructed by Anick~\cite[p.~42]{Anick-homotopyhard} for his hardness result. These are compact $4$-dimensional cell complexes which arise from the wedge $W =S_1^2\vee \cdots \vee S_r^2$ of $r$ copies of $S^2$ by attaching $s$ $4$-cells.
 According to Corollary~\ref{c:twospheres}, the homotopy class of a general attaching map has to be an integral combination of the homotopy classes $\nu_i\circ\eta$ and the Whitehead products $[\nu_i,\nu_j]$,
 where  $\eta$ is
 the homotopy class of the Hopf map $S^3\to S^2$ and $\nu_i$ is the homotopy class of the inclusion $S^2\to W$ on the $i$th copy of $S^2$.

Together with \cite[Proposition~0.18]{Hatcher} this implies that,
 up to homotopy equivalence, a completely general way of attaching
 $4$-cells to $W$ is described by integers
$a_{ij}^{(\theeqn)}$ for $1\leq i\leq j\leq r$ and $\theeqn=1,2,\ldots,s$. Specifically, the $\theeqn$th $4$-cell is attached via
a map $S_\theeqn^3 \to W$ representing  the homotopy class in
$\pi_3(W)$ defined by
\begin{equation*}
\label{e:Anick}
\varphi_\theeqn=\sum_{1\leq i\leq r} a_{ii}^{(\theeqn)} \iota_i\circ \eta
+\sum_{1\leq i<j\leq r}a_{ij}^{(\theeqn)}[\iota_i,\iota_j].
\end{equation*}
Therefore, the \emph{homotopy type} of the resulting \emph{Anick complex},
i.e., its class of homotopy equivalence, which we denote by $Y_{\aa}^4$,
is completely determined by the vector $\aa=(a_{ij}^{(\theeqn)})_{1\leq i\leq j \leq r}^{1\leq\theeqn\leq s}$ of integer coefficients.


In Anick's \#P-hardness result, the input complex $\anick{\aa}$ (whose higher homotopy groups are to be computed)
is encoded very concisely by the vector $\aa$ of integers, represented in binary:
\begin{theorem}[Anick~\cite{Anick-homotopyhard}]\label{t:anick}
It is \#P-hard to compute the rank of $\pi_\then(\anick{\aa})$ for a given integer $\then\geq 2$ (encoded in unary) and a given
integer vector $\aa$ (represented in binary).
\end{theorem}

In Section~\ref{s:poly-constructions}, we will show that, given $\aa$, we can construct, in polynomial time, a finite $4$-dimensional
simplicial complex homotopy equivalent to $\anick{\aa}$. Together with Anick's result, this will imply Theorem~\ref{t:sharpP}.

\section{Simplicial constructions}
\label{s:poly-constructions}

In this section, we prove that the constructions of cell complexes
and cellular maps
from the last section can be converted into homotopy equivalent
finite simplicial complexes and simplicial maps.
Moreover, we exhibit algorithms for constructing such
simplicial sets and maps that run in time polynomial in the
encoding size of the integer vector
$a_{ij}^{(\theeqn)}$ (and possibly $b_\theeqn$)
represented in binary.  The constructions involve only simplicial products,
attachments (in particular, mapping cylinders), quotients,  and subdivisions,
which are all algorithmic.

The polynomial bound for the running time is needed only for Anick's
space, where polynomial running time is important; for the undecidability
results we could use less efficient (and simpler) techniques.
However, there is almost no overall saving in constructing only
Anick's space with a polynomial bound and doing the other constructions
more wastefully, since all of the involved spaces are similar.
Moreover, we expect the tools developed here to be useful, e.g.,
for future NP-hardness or \#P-hardness results, where polynomiality
of the constructions is crucial, of course.

Let us denote by $\Sigma^\sphdim$ a ``model'' of the sphere $S^\sphdim$ as a simplicial set
with only two nondegenerate simplices, one in dimension $0$ and the other in dimension~$\sphdim$.

\subsection{Constructing the sum of several maps $S^\sphdim\to Y$.}
\label{s:sum}
In this short section we describe how, given simplicial
maps $f_1,\ldots,f_\thelen\:\Sigma^\sphdim\to Y$, we can construct a
simplicial representative of the sum
$[f_1]+\cdots+[f_\thelen]\in\pi_\sphdim(Y)$.
To this end,
we have to change the domain to a simplicial set with a larger number of
simplices.

We define the simplicial set $\chsimp{\sphdim}{\thelen}$, which is a union of $\thelen$ copies of $\Delta^\sphdim$, where the $i$th copy is glued by its
$\partial_1$-face to the $\partial_0$-face of the $(i+1)$st copy. The union of the remaining faces (the $\partial_1$-face of the first copy, the
$\partial_0$-face of the last copy and all the $\partial_i$-faces with $i>1$)
is denoted by $\partial \chsimp{\sphdim}{\thelen}$.
Here is a picture of $\chsimp{2}{\thelen}$:
\[\xy
(0,0)*{\scriptstyle\bullet};
(30,20)*{\scriptstyle\bullet}="a"**\dir{=}?(.5)*\dir{>},(20,0)*{\scriptstyle\bullet}**\dir{=};?(.5)*\dir{>};"a"**{}?(.33)*{\scriptstyle 1},(20,0);
"a"**\dir{-}?(.5)*\dir{>},(40,0)*{\scriptstyle\bullet}="b"**\dir{:};?(.5)*\dir{>};"a"**{}?(.33)*{\scriptstyle\cdots},"b";
"a"**\dir{-}?(.5)*\dir{>},(60,0)*{\scriptstyle\bullet}="c"**\dir{=};?(.5)*\dir{>};"a"**{}?(.33)*{\scriptstyle\thelen},"c";
"a"**\dir{=}?(.5)*\dir{>}
\endxy\]
with the double arrows denoting the boundary $\partial\chsimp{2}{\thelen}$. Another point of view is, that $\chsimp{1}{\thelen}$ is a chain of $\thelen$ copies of the $1$-simplex $\Delta^1$ and each $\chsimp{\sphdim}{\thelen}$ is a cone over $\chsimp{\sphdim-1}{\thelen}$.

There is a simplicial map $\chsimp{\sphdim}{\thelen}\to\Delta^\sphdim$ that sends the first copy of $\Delta^\sphdim$ in $\chsimp{\sphdim}{\thelen}$ onto $\Delta^\sphdim$ by the identity, while the rest is sent to the degeneracy of the $\partial_0$-face of $\Delta^\sphdim$. It induces a simplicial map
\[q\:\chsimp{\sphdim}{\thelen}/\partial \chsimp{\sphdim}{\thelen}\to\Sigma^\sphdim,\]
which is a homotopy equivalence (it is easy to see this, e.g.,
from homology). There is another simplicial map that collapses the whole $(\sphdim-1)$-skeleton of $\chsimp{\sphdim}{\thelen}$. The map factors through $\chsimp{\sphdim}{\thelen}/\partial \chsimp{\sphdim}{\thelen}$ as
\[\xymatrix{\chsimp{\sphdim}{\thelen}\ar[r]&\chsimp{\sphdim}{\thelen}/\partial \chsimp{\sphdim}{\thelen}\ar[r]^{\delta}&\Sigma^\sphdim\vee\cdots\vee\Sigma^\sphdim.} \]
We specify a simplicial map $f\:\Sigma^\sphdim\vee\cdots\vee\Sigma^\sphdim$ by mapping the $i$th copy of $\Sigma^\sphdim$ to $Y$ by $f_i$.

The maps $q$, $\delta$, and $f$ fit into a diagram
\[\xymatrix{
\chsimp{\sphdim}{\thelen}/\partial \chsimp{\sphdim}{\thelen} \ar[r]^-\delta \ar[d]_-q^-\sim & \Sigma^\sphdim\vee\cdots\vee\Sigma^\sphdim \ar[r]^-f & Y \\
\Sigma^\sphdim \ar@/_5pt/@{-->}[rru]
}\]
Since $q$ has a continuous homotopy inverse, there is a unique homotopy class of maps $\Sigma^\sphdim\to Y$ extending $[f\delta]$ up to homotopy, namely the homotopy class of $[f_1]+\cdots+[f_\thelen]$.


By the naturality of the subdivision, we also have maps
\[q\:\sdit(\chsimp{\sphdim}{\thelen}/\partial \chsimp{\sphdim}{\thelen})\to\Sigma^\sphdim,\ \ \ \delta\:\sdit(\chsimp{\sphdim}{\thelen}/\partial \chsimp{\sphdim}{\thelen})\to\sdit(\Sigma^\sphdim)\vee\cdots\vee\sdit(\Sigma^\sphdim)\]
(the map $q$ is the composition of the subdivision of the original $q$ with the iterated last vertex map $\sdit(\Sigma^\sphdim)\to\Sigma^\sphdim$) that will serve to add representatives $f_1,\ldots,f_\thelen\:\sdit(\Sigma^\sphdim)\to Y$.

\subsection{Generalized mapping cylinders}\label{s:gmc}

In the above approach, in order to construct simplicial maps $\Sigma^\sphdim\to Y$, we replaced the domain $\Sigma^\sphdim$ by a homotopy equivalent simplicial set. This will be very useful for the proof of part (a) of Theorem~\ref{t:undecide}. For part (b), the domain has to be left unchanged,
 and thus a different construction has to be used.\footnote{There is a further issue with the subdivision---it is not polynomial. The construction of a
representative of a multiple $b[f]$ of a map $f\:\Sigma^\sphdim\to Y$ requires $b$ simplices, and this number is exponential in the number of bits of~$b$.} It is
 roughly ``dual'' to the previous one: it
 replaces the target $Y$ by a homotopy equivalent simplicial set.

Thus, instead of subdividing the sphere $\Sigma^\sphdim$, we will replace the target space $Y$ by a ``generalized mapping cylinder''. This solution works also for domains other than $\Sigma^\sphdim$. Thus, for a map $f\:X\to Y$, we will be interested in diagrams, commutative up to homotopy, of the following form.
\[\xymatrix{
& Y \ar[d]^-{i_Y}_-\sim \\
X \ar[r]_-{i_X} \ar@{-->}[ru]^-f & M
}\]

\begin{defi}
Let $M$ be a pointed simplicial set with two simplicial subsets $X,Y\subseteq M$ containing the basepoint of $M$. Let $i_X\colon X\rightarrow M$ and $i_Y\colon Y\rightarrow M$ be the corresponding inclusion maps, and let $f\colon|X|\rightarrow|Y|$ be a pointed continuous map. We say that $M$ is a \emph{generalized mapping cylinder} for $f$, with \emph{upper rim} $X$ and
\emph{lower rim} $Y$, if $i_Y$ is a homotopy equivalence and $i_X\sim i_Y\circ f$. We denote this situation by $M\:X\gmc{f}Y$.
\end{defi}

We remark that the above definition depends only on the homotopy class of $f$; we may thus say that $M$ is a generalized mapping cylinder for the homotopy class $[f]$.

By Lemma~\ref{lem:mapping-cylinder} and the remark
following it, the reduced mapping cylinder $\rcyl(f)$ of $f\colon
X\rightarrow Y$ is an example of such a generalized mapping cylinder with upper rim $X$ and lower rim $Y$.
Moreover, if $f$ is a homotopy equivalence with a homotopy inverse $g$,
one can easily see from the definition that a generalized mapping cylinder $M$ for $f$ is also a generalized mapping cylinder $M^{\operatorname{op}}$ for $g$ with upper rim $Y$ and lower rim $X$ (i.e., the roles of upper and lower rim
are interchanged).\footnote{When $f$ is injective, one can use $Y$ as a generalized mapping cylinder for $f$.}

The important property of generalized mapping cylinders, that we are going to use heavily, is that they may be used for attaching cells.

\begin{prop}\label{p:gmc_cell_attach}
Let $M\:\Sigma^\sphdim_1\vee\cdots\vee\Sigma^\sphdim_\thelen\gmc fY$ be a generalized mapping cylinder for a pointed map $f$, whose restriction to the $i$th summand is $f_i\:|\Sigma^\sphdim_i|\to|Y|$. Then the composition $Y\xrightarrow{i_Y}M\xrightarrow{\operatorname{proj}}M/(\Sigma_1^\sphdim\vee\cdots\vee\Sigma_\thelen^\sphdim)$ extends to a homotopy equivalence
\[Y\cup(e_1^{\sphdim+1}\cup\cdots\cup e_\thelen^{\sphdim+1})\xrightarrow{\ \sim\ } M/(\Sigma_1^\sphdim\vee\cdots\vee\Sigma_\thelen^\sphdim)\]
where the cell $e_i^{\sphdim+1}$ on the left is attached to $Y$ along the map $f_i$.
\end{prop}

\begin{proof}
Put $X=\Sigma^\sphdim_1\vee\cdots\vee\Sigma^\sphdim_\thelen$. Then the space from the statement, obtained from $Y$ by attaching cells, is the mapping cone of $f$. By \cite[Theorem~I.14.19]{Bredon-at} the mapping cone of $f\:X\to Y$ is homotopy equivalent to that of $i_Yf\:X\to M$ (since $i_Y$ is a homotopy equivalence). Further, by \cite[Theorem~I.14.18]{Bredon-at} it is also homotopy equivalent to the mapping cone of $i_X\:X\to M$ (since $i_Yf\sim i_X$). By \cite[Theorem~VII.1.6]{Bredon-at} the mapping cone of $i_X$ is homotopy equivalent to $M/X$.

All the involved maps respect $Y$ (which is naturally a subspace of all the mapping cones and also maps by $\operatorname{proj}i_Y$ to the quotient $M/X$), proving that the resulting homotopy equivalence is indeed an extension of the composition $\operatorname{proj}i_Y$.
\end{proof}


Generalized mapping cylinders can be composed in an obvious
way:
\begin{lemma}
\label{lem:cylinder-composition} Let $f\colon X\rightarrow
Y$ and $g\colon Y\rightarrow Z$ be pointed continuous maps, and let $M$
and $N$ be generalized mapping cylinders for $f$ and $g$,
respectively. Let $N\circ M:=N\cup_Y M$ be the simplicial set
obtained by identifying the lower rim of $M$ with the upper
rim of $N$. Then $N\circ M$ is a generalized mapping
cylinder for $g\circ f$.
\end{lemma}

\begin{proof} Consider the diagram
$$
\xymatrix@R=5pt@C=20pt{
X \ar[dd]_{f} \ar[dr]^{i_X} \\
& M \ar[dr]^{i_M} \\
Y \ar[dd]_{g} \ar[ru]^{i_Y} \ar[dr]^{j_Y} & & \leftbox{N}{{}\cup_{Y}M} \\
& N \ar[ru]_{i_N} \\
Z \ar[ru]_{j_Z}
}
$$
where $i_X,\  i_Y,\ j_Y,\ i_Z,\ i_M,\ i_N$ are inclusions, both triangles commutes up to homotopy and the square commutes strictly. Consequently, the triangle formed by the spaces $X$, $Z$ and $N\cup_{Y}M$ commutes up to homotopy, too. To show that $N\cup_{Y}M$ is a generalized mapping cylinder for $gf$, it suffices to prove that the inclusion $i_N$ is a homotopy equivalence.

It is well known (see \cite[Theorem~4.5]{Hatcher}) that to every inclusion $i_Y\:Y\to M$ which is a homotopy equivalence there is a deformation retraction $r\:M\to Y$. Then the map
$h\:N\cup_{Y}M\to N$ defined as $j_Y\circ r$ on $M$ and as the identity on $N$ is a homotopy inverse to $i_N$.
\end{proof}

We will also need simplicial maps representing (the homotopy classes of) a constant number of specific maps, such as the Whitehead product $[\iota,\iota]\:S^{2d-1}\to S^d$ of the identity on $S^d$ with itself, the Whitehead product $[\iota_1,\iota_2]\:S^{2d-1}\to S^d_1\vee S_2^d$
of the inclusions $S_i^d\to S^d_1\vee S^d_2$ and the Hopf map $\eta\:S^3\to S^2$. In each case, it is possible
 to construct these explicitly, but  we will  use the following general lemma:

\begin{lemma}
\label{l:constant-gmc} Let $X$ and $Y$ be finite simplicial sets and let $f \colon |X| \rightarrow |Y|$ be an
arbitrary but fixed pointed continuous map. Then there exists a generalized mapping cylinder $X\gmc{f}Y$. It is of dimension $\max\{\dim X+1,\dim Y\}$.
\end{lemma}

The point here is that, in contrast with Theorem~\ref{t:s-apx-sset}, we can prescribe the exact triangulations of the upper and lower rim,
which will make it easy to compose the resulting generalized mapping cylinders.

\begin{proof} By the simplicial approximation theorem for simplicial sets
(Theorem~\ref{t:s-apx-sset}), there exist an iterated barycentric subdivision $X'=\sdit(X)$ of $X$ and a simplicial map $g\colon X'\rightarrow Y$
 homotopic to $f\circ|\ell|$, where $\ell\:X'\to X$
is the natural homotopy equivalence (the iterated last vertex map).

Let $M:=\rcyl(\ell)$ and $N:=\rcyl(g)$ be the corresponding reduced simplicial mapping cylinders.
Since $\ell$ is a homotopy equivalence, we can also view $M$ as a generalized mapping cylinder $M^\mathrm{op}$ for a homotopy inverse
$h\:|X|\to|X'|$ of $|\ell|$, with upper rim $X$ and lower rim $X'$.
Thus, $N\circ M^\mathrm{op}$ is a generalized mapping cylinder for $|g|\circ h\sim f\circ|\ell|\circ h\sim f$.
\end{proof}

The following proposition plays
a crucial role in our simplicial constructions:

\begin{prop}\label{p:poly-gmc} Let $Y$ be a finite simplicial set and let $f_1,\dots,f_\thelen\:\sdit\Sigma^\sphdim\to Y$ be given simplicial maps.
 Then there is an algorithm that, given an integer
vector $\cc=(c_1,\ldots,c_\thelen)$, constructs
a generalized mapping cylinder $\Sigma^\sphdim\gmc{f}Y$ for
the homotopy class
\[[f]=c_1[f_1]+\cdots+c_\thelen[f_\thelen]\in\pi_\sphdim(Y)=[\Sigma^\sphdim,Y]_*,\]
in time polynomial in the (binary) encoding size of~$\cc$.
\end{prop}

The proof will be given in a series of lemmas. Before going into the proof, we will generalize this proposition slightly. A homotopy class of a pointed map $f\:\Sigma^\sphdim_1\vee\cdots\vee\Sigma^\sphdim_s\to Y$ is determined uniquely by its restrictions $\Sigma^\sphdim_k\to Y$. When each restriction is expressed as an integral combination of the $f_i$, we may use Proposition~\ref{p:poly-gmc} together with Lemma~\ref{lem:poly-size-wedge} to provide a generalized mapping cylinder for $f$.

\begin{lemma}
\label{lem:poly-size-power-2-attachment} There is an algorithm that, given an integer of the form $c=\pm 2^d$, $d\in \N$, constructs
a generalized mapping cylinder $N_c\:\Sigma^\sphdim\gmc{}\Sigma^\sphdim$ for the map of degree~$c$.
Moreover, if $\sphdim$ is fixed then $\size(N_c)$ is linear in $d+1$ and the running time of the algorithm is polynomial in $d+1$,
which is the encoding size of $c$.
\end{lemma}

\begin{proof}
Consider maps $g_{-1},g_2\colon |\Sigma^\sphdim|\rightarrow |\Sigma^\sphdim|$ of
degrees $-1$ and $2$, respectively, and choose a generalized mapping cylinder
$N_i$ for each $g_i$ according to Lemma~\ref{l:constant-gmc}.
For $\sphdim$ constant, these are fixed simplicial sets.

Thus, by Lemma~\ref{lem:cylinder-composition},
\[N_{2^d}:=\underbrace{N_2\circ \cdots \circ N_2}_{d~\textrm{factors}}\]
is a generalized mapping cylinder for $(g_2)^d$, a
map of degree $2^d$. By further composing this with $N_{-1}$,
if necessary, we obtain a generalized mapping cylinder $N_c$
for a map of degree $c$.

Moreover, if $\sphdim$ is fixed then we can precompute the generalized mapping
cylinders $N_{-1}$, $N_2$, which leads to $\size(N_c)$ and running time as requested.
\end{proof}

\begin{lemma}
\label{lem:poly-size-wedge}
Suppose that $X_1,\ldots,X_\thelen$ and\, $Y$ are pointed simplicial sets
 and that $M_i\:X_i\gmc{}Y$ are generalized mapping cylinders for pointed maps $f_i\colon|X_i|\to|Y|$, $1\leq i\leq\thelen$.
Then there is an algorithm that constructs a generalized mapping cylinder $M\:X_1\vee \cdots \vee X_\thelen\gmc{f}Y$ for the map $f\:|X_1\vee \cdots \vee X_\thelen|\to |Y|$ with restrictions $f|_{|X_i|}=f_i$.

Moreover, if $p$ is fixed then $\size(N)$ is linear in $\sum_i\size(M_i)$ and the running time is polynomial.
\end{lemma}

\begin{proof}
The wedge sum $M'=M_1\vee\cdots\vee M_\thelen$ is a generalized mapping cylinder
\[M'\:X_1\vee \cdots \vee X_\thelen\gmc{f_1\vee \cdots \vee f_\thelen}Y\vee \cdots \vee Y.\]
We attach to $M'$ the mapping cylinder of the folding map $\nabla\:Y\vee\cdots \vee Y\to Y$ to obtain the required generalized mapping cylinder $M$.
\end{proof}

\begin{lemma}
\label{l:poly-size-sum}
Let $M_1,\ldots,M_\thelen\:\Sigma^\sphdim\gmc{}Y$ be generalized mapping cylinders for $[f_1],\ldots,[f_\thelen]\in\pi_\sphdim(Y)$. Then there is an algorithm that constructs a generalized mapping cylinder $M\:\Sigma^\sphdim\gmc{}Y$ for the homotopy class $[f_1]+\ldots+[f_\thelen]\in\pi_\sphdim(Y)$, in polynomial time if $p$ is fixed.
\end{lemma}

\begin{proof}
Let us consider the following chain of maps:
\[\Sigma^\sphdim\xleftarrow{\ q\ }\chsimp{\sphdim}{\thelen}/\partial\chsimp{\sphdim}{\thelen}\xrightarrow{\ \delta\ }\Sigma^\sphdim_1\vee\cdots\vee\Sigma^\sphdim_\thelen\xrightarrow{\ f\ }Y,\]
where $f$ restricts to $f_i$ on the $i$th summand. The first two maps are simplicial and thus their mapping cylinders provide generalized mapping cylinders for any homotopy inverse $\overline q$ of $q$ and for $\delta$, respectively. A generalized mapping cylinder for $f$ was constructed in Lemma~\ref{lem:poly-size-wedge}. Composing these cylinders gives the result, since
$f\delta\overline q\sim f_1+\cdots+f_\thelen$;
see Section~\ref{s:sum}.
\end{proof}

\begin{proof}[Proof of Proposition~\ref{p:poly-gmc}]
Let $M_i\:\Sigma^\sphdim\gmc{f_i}Y$, $i=1,\ldots,\thelen$, be generalized mapping cylinders.

Using the binary expansion of an integer $c$, Lemma~\ref{l:poly-size-sum},
 and Lemma~\ref{lem:poly-size-power-2-attachment}, we can construct the generalized mapping cylinder $N_c$ for every
map $\Sigma^\sphdim\to\Sigma^\sphdim$ of degree $c$
in time polynomial (at most quadratic)
in the bit length of $c$. The composition $M_iN_{c_i}$ is a generalized mapping cylinder for $c_i[f_i]$. Lemma~\ref{l:poly-size-sum} then constructs a generalized mapping cylinder for the sum $c_1[f_1]+\cdots+c_\thelen[f_\thelen]$.
\end{proof}

\subsection{Proofs of the main results}\label{s:proofs}

In Section~\ref{s:constructions} we described the relevant spaces as cell complexes. It remains to construct them as finite simplicial complexes and the map $f\:A\to Y$ as a simplicial map.

\heading{Proof of Theorem~\ref{t:undecide}~(a).} We will give details only for $\thedim$ even. Using the notation from Section~\ref{s:fixedtarget}, we triangulate the target sphere $Y=S^\thedim$ in an arbitrary manner and fix simplicial maps
\[w_{+},w_{-}\:\sdit(\Sigma^{2\thedim-1})\to S^\thedim\]
that represent the homotopy classes of the Whitehead square and its negative,
\[[w_\pm]=\pm[\iota,\iota]\in\pi_{2\thedim-1}(S^\thedim)\]
(by the simplicial approximation theorem, a sufficiently fine subdivision
$\sdit(\Sigma^{2\thedim-1})$ and the required simplicial maps exist,
and they can be hard-wired into the algorithm).
Let now $\bb$ be the vector of the right hand sides of an arbitrary system of the form \eqref{e:qsym} and let $1\leq\theeqn\leq s$. By adding $|b_\theeqn|$ times the map $w_\pm$, we obtain a simplicial map
\[A'_\theeqn:=\sdit\big(\chsimp{2\thedim-1}{|b_\theeqn|}/\partial \chsimp{2\thedim-1}{|b_\theeqn|}\big)\xrightarrow{\ f'_\theeqn\ }S^\thedim\]
that represents $b_\theeqn[\iota,\iota]$.\footnote{When $b_\theeqn=0$, we take $A'_\theeqn=\sdit(\Sigma^{2\thedim-1})$ and $f'_\theeqn\:\sdit(\Sigma^{2\thedim-1})\to S^\thedim$ the constant map onto the basepoint.} Finally,
 we take $A'=A'_1\vee\cdots\vee A'_s$ and specify $f'\:A'\to S^\thedim$ by its restrictions to the $A'_\theeqn$, namely, the maps $f'_\theeqn$.

We recall that the space $X$ is constructed as the mapping cylinder of a map $g\:A\to W$ that was expressed in terms of the Whitehead products $[\nu_i,\nu_j]$ and the coefficients of the system \eqref{e:qsym}.
In the simplicial setup it will be more convenient to use generalized
mapping cylinders for this purpose. As explained during the discussion of the generalized extension problem
in the beginning of Section~\ref{s:constructions}, the extension problems are equivalent.
Using a fixed representatives
$w'_{\pm}:\sdit(\Sigma^{2\thedim-1})\to\Sigma^\thedim\vee\Sigma^\thedim$,
we may construct the generalized mapping cylinder $X'\:A'\gmc{}W$,
 with the inclusion denoted by $i'\:A'\to X'$, of the composition
\[\underbrace{A'_1\vee\cdots\vee A'_s}_{A'} \xrightarrow{\ q\vee\cdots\vee q\ }
\underbrace{\Sigma^{2\thedim-1}_1\vee\cdots\vee\Sigma^{2\thedim-1}_s}_A \xrightarrow{\ g\ }
\underbrace{\Sigma^\thedim_1\vee\cdots\vee\Sigma^\thedim_r}_W.\]
Thus, we have constructed an extension problem, given by $i'$ and $f'$, and
by Proposition~\ref{p:fixed-target-k-even}, its solvability is equivalent to the solvability of the system \eqref{e:qsym} that we started with.

Finally, we replace the simplicial sets $A'$ and $X'$ by the simplicial complexes $B_*(\sd(A'))$ and $B_*(\sd(X'))$
(see Proposition~\ref{p:doublesubdiv}).
The map $f'$ is replaced by the composition $f'\gamma_{A'}$ in the diagram
\[\xymatrix{
B_*(\sd(A'))\ar[r]^-{\gamma_{A'}}\ar@{c->}[d] &A'\ar[r]^{f'}\ar@{c->}[d]&S^\thedim\\
B_*(\sd(X'))\ar[r]^-{\gamma_{X'}}& X'&
}\]
where the maps $\gamma_{A'}$ and $\gamma_{X'}$ were also defined in Proposition~\ref{p:doublesubdiv}.

Since both $\gamma_{X'}$ and $\gamma_{A'}$ are homotopy equivalences, the extendability of $|f'|$ is equivalent to that of $|f'\gamma_{A'}|$
by Corollary~\ref{c:ext}.
\proofend

\heading{Proof of Theorem~\ref{t:undecide}~(b).}
Again, we work out the case $\thedim$ even. Let $A$, $X$, $Y$, $f$ and $g$ be as in Proposition~\ref{p:fixed-target-k-even} and fix an arbitrary pair of simplicial sets $(X',A')$ whose geometric realization is homotopy equivalent to $(X,A)$. Using generalized mapping cylinders for this purpose, we may assume that $A'=\Sigma^{2\thedim-1}$. We fix some simplicial maps
\[w'_\pm\:\sdit(\Sigma^{2\thedim-1})\to\Sigma^\thedim\vee\Sigma^\thedim\]
representing the Whitehead product and its negative. According to Proposition~\ref{p:gmc_cell_attach}, the cell complex $Y$ of \eqref{e:target} is homotopy equivalent to the quotient $M/S$, where $M$ is an arbitrary generalized mapping cylinder $M\:S\gmc{\varphi}T$ for the map $\varphi\:|S|\to|T|$ between the geometric realizations of the simplicial sets
\[S=\bigvee\nolimits_{i<j}\Sigma_{ij}^{2\thedim-1}\vee\bigvee\nolimits_{i}
\Sigma_{ii}^{2\thedim-1},\ \ \ T=\bigvee\nolimits_i\Sigma^\thedim_i\vee\bigvee\nolimits_\theeqn
\Sigma^{2\thedim-1}_\theeqn,\]
whose restrictions to the spheres of $S$ are the attaching maps $\varphi_{ij}$ and $\varphi_{ii}$ for the cells of $Y$; see Section~\ref{s:constructions}. The generalized mapping cylinder $M$ is constructed by Proposition~\ref{p:poly-gmc}.

Since the image of $f\:A\to Y$ lies in $T$, the replacement of $Y$ by $M/S$ results in replacing $f$ by the composition
\[\xymatrix{
\tilde f\:\Sigma^{2\thedim-1} \ar[r]^-f & T \ar@{ ->}[r]^-{i_{T}} & M \ar[r]^-{\operatorname{proj}} & M/S.
}\]
(the homotopy equivalence $Y\simeq M/S$ restricts to $\operatorname{proj}i_T$ on $T$ by Proposition~\ref{p:gmc_cell_attach}).

It remains to replace $\tilde f$ by a simplicial map. But since $f$ is a combination of the Whitehead products and the remaining maps $i_T$ and $\operatorname{proj}$ are simplicial, we may achieve this by replacing $M/S$ further by the generalized mapping cylinder $Y'\:\Sigma^{2\thedim-1}\gmc{\tilde f}M/S$ as in Section~\ref{s:gmc}. We denote the inclusion $\Sigma^{2\thedim-1}\to Y'$
by $f'$. From the definition of the generalized mapping cylinder, $i_{M/S} \tilde f\sim f'$ and $i_{M/S}$ is a homotopy equivalence, and therefore the extendability of $\tilde f$ is equivalent to that of $f'$. This finishes the construction of a simplicial replacement of the extension problem.

To make everything into simplicial complexes, we apply $B_*\sd$ to all the involved simplicial sets $A'=\Sigma^{2\thedim-1}$, $X'$, $Y'$ and the simplicial map $f'$.\proofend

\heading{Proof of Theorem~\ref{t:sharpP}.}
Let us fix some simplicial representatives
\[h\:\sdit(\Sigma^3)\to\Sigma^2,\ \ \ w\:\sdit(\Sigma^3)\to\Sigma^2\vee\Sigma^2,\ \ \ m_2,m_{-1}\:\sdit(\Sigma^2)\to\Sigma^2\]
for the Hopf map $\eta$, the Whitehead product $[\iota_1,\iota_2]$, and the maps of degree $2$ and $-1$, respectively. We may then build the generalized mapping cylinder
\[M\:\Sigma^3_1\vee\cdots\vee\Sigma^3_s\gmc{\varphi_\aa}\Sigma^2_1\vee\cdots\vee\Sigma^2_r,\]
for the map $\varphi_\aa$ whose restriction to the $k$th summand $\Sigma^3_k$ is given by
\[\varphi_\theeqn=\sum_{1\leq i\leq r} a_{ii}^{(\theeqn)} \iota_i\circ \eta+\sum_{1\leq i<j\leq r}a_{ij}^{(\theeqn)}[\iota_i,\iota_j].\]
We construct  Anick's simplicial complex $\anick{\aa}$ as
\[\anick{\aa}=B_*\sd\big(M/(\Sigma^3_1\vee\cdots\vee\Sigma^3_s)\big).\]
By Proposition~\ref{p:gmc_cell_attach}, it is homotopy equivalent to the
cellular complex obtained from the wedge $\Sigma^2_1\vee\cdots\vee\Sigma^2_r$
by attaching $4$-cells along the maps with homotopy classes~$\varphi_\theeqn$.

To get the statement of  Theorem~\ref{t:sharpP}, it is now sufficient to realize that the algorithmic construction above can be carried out  in the time polynomial in the binary encoding of the vector $a_{ij}^{(\theeqn)}$, $1\leq i\leq j\leq r$, $1\leq \theeqn\leq s$.\proofend

\heading{Acknowledgement.} We are grateful to Maurice Rojas
for providing us useful information on Diophantine equations.

\appendix
{\renewcommand\thesection{\appendixname\ \Alph{section}.}
\section{Converting a simplicial set into a simplicial complex}}
\label{a:subdivisions}
Here we outline a proof of Proposition~\ref{p:doublesubdiv}.

Since a simplicial set also encodes an ordering of the vertices within each simplex, there is
another ``barycentric subdivision'' $\sd_*(X)$ associated  with any simplicial set $X$, obtained
by reversing the order of the vertices in every simplex of $\sd(X)$ (in $\sd(X)$, the inclusion chains
of simplices are ordered according to ascending dimension, and in $\sd_*(X)$ according to descending dimension).
Thus, for example, $\sd(\Delta^1)$ can be described pictorially as $\xymatrix@1{\bullet \ar[r] & \bullet & \bullet \ar[l]}$
while $\sd_*(\Delta^1)$ is $\xymatrix@1{\bullet & \bullet \ar[l] \ar[r] & \bullet}$.
 The barycentric subdivision $\sd_*(X)$ is related to the original
simplicial set $X$ via an initial vertex map $\sd_*(X)\to X$,
which is a homotopy equivalence.

Moreover there is a universal way of associating a simplicial complex with any simplicial set $X$: it has the same vertex
set as $X$ and a collection of vertices spans a (unique) simplex if and only if there \emph{exists} a simplex in the original simplicial set $X$ with this vertex set. An alternative, equivalent definition of $B_*X$ is that it is the simplicial complex associated in this way
with $\sd_*(X)$.

\begin{proof}[Proof of Proposition~\ref{p:doublesubdiv}]
The face operators $\partial_i$ can be iterated to obtain more general face operators. Since each $\partial_i$ leaves out the $i$-th vertex of a simplex, by iterating we obtain face operators that leave out a set of vertices. When this set is $I\subseteq\{0,\ldots,n\}$, we write the corresponding operator as
$\partial_I$. It is easy to observe that we can express $\partial_I$ as
\[\partial_I=\partial_{i_1}\cdots\partial_{i_k},\]
where $i_1<\cdots<i_k$ is the ordering of the elements of $I=\{i_1,\ldots,i_k\}$. We call $k$ the \emph{codimension}
 of $\partial_I$. Similarly, we can iterate the $s_i$ and obtain general
degeneracy operators~$s_I$.


Since we are interested in the computational side of the story, we will describe the simplicial complex $B_*(\sd(X))$ explicitly. Its vertices are the chains
\[\xymatrix{
\bs=\sigma_0 \ar@{ >->}[r]^-{f_1} & \sigma_1 \ar@{ >->}[r]^-{f_2} & \ar@{.}[r] & \ar@{ >->}[r]^-{f_k} & \sigma_k,
}\]
where $\sigma_0,\ldots,\sigma_k$ are simplices of $X$ with $\sigma_k$ non-degenerate, and each $f_i$ is a face operator of codimension at least $1$, for which $\sigma_{i-1}=f_i\sigma_i$.

We say that a chain
\[\xymatrix{
\bt=\tau_0 \ar@{ >->}[r]^-{g_1} & \tau_1 \ar@{ >->}[r]^-{g_2} & \ar@{.}[r] & \ar@{ >->}[r]^-{g_\ell} & \tau_\ell
}\]
(still with $\tau_\ell$ non-degenerate and all face operators $g_i$ of codimension at least $1$) is a subchain of $\bs$, which we write as $\bs>\bt$,
if there exists an injective monotone map (a subsequence) $\varphi\:\{0,\ldots,\ell\}\to\{0,\ldots,k\}$ with $\ell<k$ and a commutative diagram
\[\xymatrix{
\sigma_{\varphi(0)} \ar@{ >->}[r] \ar@{->>}[d] & \sigma_{\varphi(1)} \ar@{ >->}[r] \ar@{->>}[d] & \ar@{.}[r] & \ar@{ >->}[r] & \sigma_{\varphi(\ell)} \ar@{->>}[d] \\
\tau_0 \ar@{ >->}[r]^-{g_1} & \tau_1 \ar@{ >->}[r]^-{g_2} & \ar@{.}[r] & \ar@{ >->}[r]^-{g_\ell} & \tau_\ell
}\]
(the top maps are the appropriate compositions of the $f_i$), where every arrow $\xymatrix@1@C=15pt{\sigma_{\varphi(i)} \ar@{->>}[r] & \tau_i}$ is an iterated degeneracy operator $p_i$ for which $p_i\tau_i=\sigma_{\varphi(i)}$. The commutativity means that the respective compositions of operators are equal. For each $\varphi$, there exists at most one subchain $\bt$, but for a given $\bt$, the choice of $\varphi$ is not unique. The composition $\sigma_0\to\sigma_{\varphi(0)}\to\tau_0$ gives a canonical operator $\sigma_0\to\tau_0$. It is not too hard to show\footnote{The important ingredients are that every simplex can be expressed uniquely as a degeneracy of a non-degenerate simplex, and that every operator can be written uniquely as a degeneracy of a face.} that this operator depends only on $\bs$ and $\bt$ and not on the choice of $\varphi$.


The $n$-simplices of $B_*(\sd(X))$ are then formed by the subsets $\{\bs^0,\ldots,\bs^n\}$ for all decreasing sequences $\bs^0>\cdots>\bs^n$ of chains; we order the vertices in each simplex according to the subchain relation.

The simplicial map $\gamma\:B_*(\sd(X))\to X$ is defined on vertices by sending $\bs$ to the last vertex of $\sigma_0$. For a simplex specified by $\bs^0>\cdots>\bs^n$, we have a canonical chain
\[\xymatrix{
(\sigma^0)_0 \ar[r] & (\sigma^1)_0 \ar[r] & \ar@{.}[r] & \ar[r] & (\sigma^n)_0
}\]
of operators and we use these to map the last vertices $\lastvertex((\sigma^i)_0)$ of the faces $(\sigma^i)_0$ to $(\sigma^n)_0$. In this way, we obtain an (ordered) collection of vertices of $(\sigma^n)_0$. The value of $\gamma$ on the sequence $\bs^0>\cdots>\bs^n$ is then the simplex of $(\sigma^n)_0$ spanned by these vertices (it might be degenerate, e.g.~when some of the $(\sigma^{i-1})_0\to(\sigma^i)_0$ preserve the last vertex).

According to \cite[Prop.~4.6.3]{FritschPiccinini:CellularStructures-1990} and \cite[Cor.~4.3]{Jardine:SimplicialApproximation-2004}, the horizontal map and the vertical map in the triangle
\[\xymatrix{
\sd_*(\sd(X)) \ar[r]^-\pi \ar[d] & B_*(\sd(X)) \ar[ld]^-\gamma \\
X
}\]
are homotopy equivalences (the vertical map is the composition of the initial vertex map with the last vertex map). Since the diagram commutes,%
\footnote{This is not too hard to show, but we do not want to dwell into the exact definition of $\sd_*(\sd(X))$. The main point is that the preimages under $\pi$ of the simplex $\bs^0>\cdots>\bs^n$ are given by the choices of the subsequences $\varphi$. The commutativity is then implied by the independence of the operators $(\sigma^{i-1})_0\to(\sigma^i)_0$ on these choices.}
 the map $\gamma$ must be a homotopy equivalence, too.
%
\end{proof}

{\renewcommand\thesection{\appendixname\ \Alph{section}.}
\section{Extending maps into $(\thedim-1)$-reduced simplicial sets}}
\label{a:reduced}
Here we prove the claim made after Theorem~\ref{t:undecide}
 regarding the construction of the target space $Y$
 as a $(k-1)$-reduced simplicial set.
It is usual in effective algebraic topology that certain computations with simplicial sets only work when at least some of the inputs are $0$-reduced or $1$-reduced. A typical example is the computation of the homology groups of a loop space $\Omega X$ of a simplicial set $X$. First, we remark, that it is impossible to compute these homology groups for general $X$, as otherwise we would obtain an algorithmic computation of $\pi_1(X)$, which is known to be impossible by (cite). On the positive side, there is a very old method for the computation of these homology groups which, however, works only for $1$-reduced $X$ (using the so-called cobar construction and homological perturbation theory; see (cite)).\footnote{On the other hand it is possible to compute these homology groups for any $1$-connected $X$. This is another application of the ``brute force'' version of Theorem~\ref{t:reduced_replacement} that is explained just prior to its statement.}

It is thus natural to ask if the undecidability of the extension problem of Theorem~\ref{t:undecide} might only be caused by $Y$ being $(\thedim-1)$-connected but not $(\thedim-1)$-reduced. In this section, we will prove a version of Theorem~\ref{t:undecide} with $(\thedim-1)$-reduced $Y$. The simplicial set $Y'$ appearing in the proof of part (a) of Theorem~\ref{t:undecide} might be chosen to be either $\Sigma^\thedim$ or $\Sigma^\thedim\vee\Sigma^\thedim$ (depending on the dimension), both of which are $(\thedim-1)$-reduced and no further work is needed.

To finish the proof of part (b), we need to replace $Y'$ by some $(\thedim-1)$-reduced simplicial set without changing the extendability. To this end, we introduce a very useful notion of an $n$-equivalence. Let $0\leq n\leq\infty$. A continuous map $Y\to Z$ is said to be an \emph{$n$-equivalence} if it induces an isomorphism on all homotopy groups up to dimension $n-1$ and a surjection in dimension $n$ (the $\infty$-equivalences are usually called weak homotopy equivalences). We will need the following basic property of $n$-equivalences.

\begin{prop}
Let $X$ be a cell complex of dimension $n$, $f\:A\to Y$ a continuous map defined on a subcomplex $A\subseteq X$ and $h\:Y\to Z$ an $n$-equivalence. Then $f$ is extendable to $X$ if and only if $hf$ is extendable to $X$.
\end{prop}

\begin{proof}
If $f$ has an extension $g$, then $hf$ extends to $hg$. The other direction is \cite[Section~7.6, Theorem~22]{Spanier}.
\end{proof}

With the previous proposition in mind, we construct a replacement of $Y$ by brute force, i.e.~by going through all ``candidate replacements'' and checking if they give equivalent extension problems. In detail, we fix $(X',A')$ as in Section~\ref{s:proofs} and make a list of all pairs $(Z',h')$, where $Z'$ is a finite $(\thedim-1)$-reduced simplicial set and $h'\:Y'\to Z'$ a simplicial map. In each step we test whether $h'$ is a $(2\thedim)$-equivalence. If that is the case,
 then the problem of extending the composition $h'f'\:A'\to Z'$ to $X'$ is equivalent to that of $f'\:A'\to Y'$, which we proved to be undecidable.

By the Hurewicz theorem, $h'$ is a $(2\thedim)$-equivalence if and only if $\Cone(h')$ has zero homology groups up to dimension $2\thedim$. This can be tested easily using a Smith normal form algorithm. The pair $(Z',h')$ with the above properties exists by the following theorem, finishing the proof of Theorem~\ref{t:undecide} with $(\thedim-1)$-reduced target.

\begin{theorem}\label{t:reduced_replacement}
Let $n\geq\thedim\geq 2$ and let $Y$ be a $(\thedim-1)$-connected simplicial set whose homology groups $H_i(Y)$, $i\leq n$, are finitely generated. Then there exist a finite $(\thedim-1)$-reduced simplicial set $Z$ and an $n$-equivalence $\psi\colon Y\to Z$.
\end{theorem}

We believe that the theorem should be known in some form but we were not able to find it in the literature. For its proof, we will need a couple of more advanced notions. Accordingly, the proof will be more sketchy.

Let $\Delta^n$ be the standard $n$-simplex (regarded as a simplicial set).
For $0\leq k\leq n$, the \emph{$k$th $n$-horn} is the simplicial subset $\Lambda^n_k\subseteq \Delta^n$
spanned by all the proper faces of $\Delta^n$ with the exception of the $k$-th face.

A simplicial set $Z$ is said to be a \emph{Kan complex} if every simplicial map $f\:\Lambda^n_k\to Z$ can be extended to a simplicial map $\Delta^n\to Z$. The map $f$ is called a \emph{horn in $Z$} and we say that this horn \emph{can be filled} if the extension exists.

A usual method for constructing Kan complexes is the \emph{successive
filling of horns}, which is described, e.g., in 
\cite[Proofs of Prop.~4.5.5 and 4.5.6]{FritschPiccinini:CellularStructures-1990}.

Given a simplicial set $Y$ and a horn $f\colon
\Lambda^n_k\rightarrow Y$ in $Y$,
we can form a larger simplicial set $Y\cup_f \Delta^n$ by
attaching $\Delta^n$ to $Y$ along $f$; this larger simplicial
set (continuously) deformation retracts to $Y$, and, by construction, the
horn $f\colon \Lambda^n_k\rightarrow Y$ can be filled in the
larger simplicial set. We refer to this operation as a \emph{single horn filling}.

If we simultaneously attach fillings for all unfillable horns
in $Y$, we obtain a simplicial set $\Ex(Y)$ that contains $Y$
and such that all horns in $Y$ can be filled in $\Ex(Y)$.
Iterating this procedure\footnote{Formally, we define $\Ex^{0}(Y):=Y$
and $\Ex^{n}(Y):=\Ex(\Ex^{n-1}(Y))$ for $n\geq 1$.}, we
obtain a sequence $Y \subseteq \Ex(Y)\subseteq \Ex^2(Y)\subseteq \ldots$, where, by construction,
every horn in $\Ex^{n}(Y)$, can be filled in $\Ex^{n+1}(Y)$.
Let $\Ex^\infty(Y)$ be the union of the simplicial sets
$\Ex^n(Y)$, $n\in \N$.

\begin{prop}\hfill
\begin{itemize}
\item
The simplicial set $\Ex^\infty(Y)$ is a Kan complex.

\item
The inclusion $Y\to\Ex^\infty(Y)$ is an $\infty$-equivalence.

\item
If $Y$ is $(d-1)$-reduced then so is $\Ex(Y)$ and consequently also $\Ex^\infty(Y)$.

\item
Let $L\subseteq\Ex^\infty(Y)$ be a finite simplicial subset. Then $L$ lies in a simplicial subset obtained from $Y$ by a finite sequence of single horn fillings.
\end{itemize}
\end{prop}

\begin{proof}
Any horn in $\Ex^\infty(Y)$ lies in some $\Ex^n(Y)$ and is thus fillable in $\Ex^{n+1}(Y)\subseteq\Ex^\infty(Y)$. The proof of the second point is similar, using a deformation retraction of $\Ex^n(Y)$ onto $Y$. The third point is clear. For the last point, $L$ lies in some $\Ex^n(Y)$. By induction, $L\cap\Ex^{n-1}(Y)$ uses only a finite number of single horn fillings. A finite number of single horn fillings is required to cover $L$.
\end{proof}

Thus, there exists a homotopy equivalence $Y\to\Ex^\infty(Y)$ of any simplicial set with a Kan complex. For the proof of Theorem~\ref{t:reduced_replacement}, we will need that every Kan complex contains a minimal Kan complex as a deformation retract \cite[Theorem~9.5]{May:SimplicialObjects-1992} and that a minimal $(\thedim-1)$-connected Kan complex is in fact $(\thedim-1)$-reduced. Thus, composing the inclusion $Y\to\Ex^\infty(Y)$ with the deformation retraction yields an $\infty$-equivalence $\iota\:Y\to Z$ of $Y$ with a $(\thedim-1)$-reduced Kan complex $Z$.

\heading{Proof of Theorem~\ref{t:reduced_replacement}. }
Let $Z^0\subseteq Z$ denote the image of $\iota$ and $\iota_0\:Y\to Z^0$ the restriction of $\iota$. By the following proposition, there exists a simplicial set $Z'$, containing $Z^0$ as a subset, and a simplicial $(n+1)$-equivalence $\psi\:Z'\to Z$. Since this map is the identity on $Z^0$, there is a canonical map $h\:Y'\to Z'$ making the diagram
\[\xymatrix@C=20pt@R=10pt{
& & Z' \ar[dd]^-\psi \\
Y \ar[r]^-{\iota_0} \ar@/^10pt/@{-->}[rru]^-{h'} \ar@/_10pt/[rrd]_-\iota & Z^0 \ar@{c->}[ru] \ar@{c->}[rd] \\
& & Z
}\]
commutative (namely the composition of $\iota_0\:Y\to Z^0$ with the inclusion $Z^0\embed{}Z'$). Since $\iota$ is an $\infty$-equivalence and $\psi$ induces an isomorphism on homotopy groups up to dimension $n$, the same is true for $h'$ and in particular, it is an $n$-equivalence.\proofend

\begin{prop}\label{p:finiteReduced}
Let $n\geq\thedim\geq 2$ and let $Z$ be a $(\thedim-1)$-connected Kan complex whose homology groups $H_i(Z)$, $i\leq n$, are finitely generated. Then there exist a finite $(\thedim-1)$-reduced simplicial set $Z'$ and an $n$-equivalence $\psi\colon Z'\to Z$.

If $Z^0$ is an arbitrary finite $(\thedim-1)$-reduced subset of $Z$, then $Z'$ can be chosen to contain $Z^0$ as a subset $\psi$ to be the identity on $Z^0$.
\end{prop}

To prove Proposition~\ref{p:finiteReduced}, we follow the
argument in \cite[Proposition~4C.1]{Hatcher} but make the
attachment maps simplicial by using the idea of filling
horns described above.

\begin{proof}[Proof of Proposition~\ref{p:finiteReduced}]
We proceed by induction on $n$. For $n\leq\thedim-1$, we can take $Z'=Z^0$.

Assume that we have constructed a finite simplicial set $Z^{n-1}$
and a map $\psi^{n-1}\colon Z^{n-1}\rightarrow Z$ that is an $(n-1)$-equivalence.

Let $\hat Z$ be the simplicial set obtained from the simplicial mapping cylinder $\Cyl(\psi^{n-1})$ of $\psi^{n-1}\:Z^{n-1}\to Z$ by collapsing $Z^0\times\Delta^1$ onto the base $Z^0$. Since $\psi^{n-1}$ is the identity on $Z^0$, the usual deformation retraction of $\Cyl(\psi^{n-1})$ onto $Z$ induces a deformation retraction of $\hat Z$ onto $Z$. We
enlarge the pair $(\hat Z,Z^{n-1})$ to a Kan pair $(K,L)$ by
filling horns. Since $Z^{n-1}$ is $(\thedim-1)$-reduced, so is $L$.

By the assumption on $\psi^{n-1}$ and by the
Hurewicz theorem, we have $H_i(K,L)=0$ for
$i\leq n-1$. Consider the exact sequence of homology groups
for the pair $(K,L)$:
$$\ldots \to H_n(K) \to H_{n}(K,L) \to H_{n-1}(L) \to H_{n-1}(K)\to H_{n-1}(K,L)=0.$$
Pick  generators $\gamma_j$ of $H_n(K,L)$. Since $H_n(K)\cong H_n(Z)$ and $H_{n-1}(L)\cong H_{n-1}(Z^{n-1})$ are finitely generated, a finite number of generators suffices. By the relative Hurewicz theorem,
$H_n(K,L)\cong \pi_n(K,L,\ast)$ (simplicial homotopy groups,
since we are working with a Kan pair). Thus, we can choose
$n$-simplices $g_j$ of $K$ representing the
$\gamma_j$, whose faces lie in $L$.

Let $K'$ be the simplicial subset of $K$ spanned by $L$ and the simplices $g_j$. Then the natural homomorphism $H_n(K',L)\to H_n(K,L)$ is surjective by the choice of the simplices $g_j$. Since these relative homology groups are zero in lower dimensions, it follows from the long exact sequence of the triple $(K,K',L)$ that $H_i(K,K')=0$ for $i\leq n$. In effect, the inclusion $K'\to K$ is an $n$-equivalence (by the relative Hurewicz theorem again). Composing with an arbitrary deformation retraction of $K$ onto $Z$ we obtain an $n$-equivalence $K'\to Z$ satisfying all the required properties except that $K'$ is infinite.

Thus, it remains to replace $K'$ by a finite simplicial set. Since there are only finitely many
simplices $g_j$, there is a simplicial set
$L'$ with $Z^{n-1} \subseteq L' \subseteq L$, obtained
from $Z^{n-1}$ by filling finitely many horns, and such that
the boundaries of all the simplices $g_j$ lie in $L'$. We take $Z^n$ to be the finite simplicial set spanned by $L'$ and the simplices $g_j$. Since $|L'|$ is a deformation retract of $|L|$, we get that
 $|Z^n|$ is a deformation retract of $|K'|$.
Thus, $Z^n$ has all the required properties.
\end{proof}

\bibliographystyle{alpha}
\bibliography{Postnikov}

\end{document}